%% file: SparseDP_journal.tex
\documentclass[runningheads]{llncs}

\usepackage{amssymb,amstext}
\usepackage{cite,scrtime}
\usepackage{graphics}
\usepackage{amsmath,dsfont}
\usepackage{mathrsfs}
\usepackage{fancyhdr}
\usepackage{verbatim}
\usepackage{boxedminipage}
\usepackage{colortbl}
\usepackage{algorithm}
\usepackage[noend]{algorithmic}
\usepackage{tikz}

\usepackage[T1]{fontenc}
\usepackage{tikz}
\usepackage{graphicx}

\tikzstyle{place}=[
  circle,
  fill=black,
  thick,
  inner sep=0.1cm,
  minimum size=2mm
]

\tikzstyle{place2}=[
  circle,
  fill=black,
  thick,
  inner sep=0.2cm,
  minimum size=2mm
]

\newcommand{\degs}[2]{\text{deg}_{#1}(#2)}

\newcommand{\tw}{{\mathbf{tw}}}
\newcommand{\pw}{{\mathbf{pw}}}

\usepackage{graphicx}

\newcommand{\bw}{{\mathbf{bw}}}
\newcommand{\mids}{{\bf mid}}

\renewcommand{\leq}{\leqslant}
\renewcommand{\geq}{\geqslant}

\usepackage{vmargin}
\setmarginsrb{3.8cm}{1.2cm}{3.8cm}{3.cm}{1.2cm}{1.2cm}{0.8cm}{1.4cm}

\def\debut{\begin{itemize}\item[{\bf [[}]\small}
\def\term{\hfill {\bf]]} \end{itemize} }

\newcommand{\NN}{\mathcal{N}}
\newcommand{\N}{\mathbb{N}}

\newtheorem{claimN}{Claim}

\renewenvironment{proof}[1][]{\par \noindent {\bf Proof:#1}\ }{\hfill$\Box$}

\usepackage{xspace}

\graphicspath{{./}}

\newcommand{\probl}[3]{
\begin{flushleft}
\fbox{
\begin{minipage}{13cm}
\noindent {\sc #1}\\
          {\bf Input:} #2\\
          {\bf Question:} #3
\end{minipage}}
\medskip
\end{flushleft}
}
\newcommand{\probls}[4]{
\begin{flushleft}
\fbox{
\begin{minipage}{13cm}
\noindent {\textsc {#1}}\\
          {\bf Input:} #2\\
          {\bf Parameter:} #4\\
          {\bf Question:} #3
\end{minipage}}
\medskip
\end{flushleft}
}

\newcommand{\RR}{\mathcal{R}}

\newcommand{\es}{\emptyset}
\newcommand{\bs}{\backslash}
\newcommand{\set}[1]{V[#1]}
\newcommand{\true}{true}

\newcommand{\tikzb}{\begin{tikzpicture}
    [every label/.style={black},place/.style={circle,fill=black,thick,inner sep=0.1cm,minimum size=2mm}]
}
\newcommand{\tikze}{\end{tikzpicture}}

\newcommand{\nod}[2]{\node at (#1,#2) [place]{}}
\newcommand{\nnod}[4]{\node (#3) at (#1,#2) [place,label=#4:#3]{}}

\newcommand{\nodd}[1]{node [place,#1] {}}
\newcommand{\noddd}[2]{node [fill=white,draw,double,rounded corners,#1] {#2}}
\newcommand{\capt}[2]{\draw[white] (#1) -- node[above,sloped,black] {#2} +(2,0)}

\title{The role of planarity in connectivity problems parameterized by treewidth\thanks{This work was supported by the ANR French project AGAPE (ANR-09-BLAN-0159) and the Languedoc-Roussillon Project ``Chercheur d'avenir'' KERNEL.}
}
\author{Julien Baste and  Ignasi Sau}

\authorrunning{Julien Baste and  Ignasi Sau}
\titlerunning{The role of planarity in connectivity problems parameterized by treewidth}

\institute{AlGCo project-team, CNRS, LIRMM, Montpellier, France.\\
\email{jbaste@ens-cachan.fr, ignasi.sau@lirmm.fr}}

\begin{document}

\maketitle
\setcounter{footnote}{0}

\vspace{-.35cm}
\begin{abstract} %In the last decade the topic of algorithms parameterized by treewidth for connectivity problems has attracted a considerable attention.
For some years it was believed that for ``connectivity'' problems such as \textsc{Hamiltonian Cycle}, algorithms running in time $2^{O(\tw)}\cdot n^{O(1)}$  --called \emph{single-exponential}--  existed only on planar and other sparse graph classes, where $\tw$ stands for the treewidth of the $n$-vertex input graph. This was recently disproved by Cygan \emph{et al}.~[FOCS 2011], Bodlaender \emph{et al}.~[ICALP 2013], and Fomin \emph{et al}.~[SODA 2014], who provided single-exponential algorithms on general graphs for essentially all connectivity problems that were known to be solvable in single-exponential time on sparse graphs. In this article we further investigate the role of planarity in connectivity problems parameterized by treewidth, and convey that several problems can indeed be distinguished according to their behavior on planar graphs. Known results from the literature imply that there exist problems, like \textsc{Cycle Packing}, that {\sl cannot} be solved in time $2^{o(\tw \log \tw)} \cdot n^{O(1)}$ on general graphs but that can be solved in time $2^{O(\tw)} \cdot n^{O(1)}$ when restricted to planar graphs. Our main contribution is to show that there exist problems that can be solved in time $2^{O(\tw \log \tw)} \cdot n^{O(1)}$ on general graphs but that {\sl cannot} be solved in time $2^{o(\tw \log \tw)} \cdot n^{O(1)}$ even when restricted to planar graphs. Furthermore, we prove that  \textsc{Planar Cycle Packing} and \textsc{Planar Disjoint Paths} cannot be solved in time $2^{o(\tw)} \cdot n^{O(1)}$.
 %and that  \textsc{Planar Subgraph Isomorphism} cannot be solved in time $2^{o(\tw \log \tw)}\cdot %n^{O(1)}$. 
The mentioned negative results hold unless the ETH fails. We feel that our results constitute a first step in a subject that can be further exploited.

%and we present several avenues for further research.

\vspace{0.25cm} \textbf{Keywords:} parameterized complexity, treewidth, connectivity problems, single-exponential algorithms, planar graphs, dynamic programming.
\end{abstract}

%\vspace{.25cm}

%\begin{center}
%\red{\fbox{Draft of \today~(\thistime h) by Ignasi}}
%\end{center}

%\vspace{.25cm}

%-------------------------------------------------------------------------------------------------------------------------------

\section{Introduction}
\label{sec:intro}
\input{intro_journal}
\section{Notation and preliminaries}
\label{sec:prelim}
\input{preliminaries}

%%-------------------------------------------------------------------------------------------------------------------------------
%\section{Tight problems of Type 1 -- 3-Colorability}
%\label{sec:type1}
%\input{type1}

%-------------------------------------------------------------------------------------------------------------------------------
\section{Problems of Type 2}
\label{sec:type2}
\input{type2}
%\input{type2mdfvs}

%-------------------------------------------------------------------------------------------------------------------------------
\section{Problems of Type 3}
\label{sec:type3}
\input{type3}

%-------------------------------------------------------------------------------------------------------------------------------
\section{Lower bound for Planar Disjoint Paths }
\label{sec:pdp}
\input{disjointPaths}

%-------------------------------------------------------------------------------------------------------------------------------
%\section{Lower bound for Planar Subgraph Isomorphism}
%\label{sec:otwlogtw}
%\input{otwlogtw}

%-------------------------------------------------------------------------------------------------------------------------------
%\section{Further research}
%\label{sec:further}
%\input{further}

%-------------------------------------------------------------------------------------
{\small
\bibliographystyle{abbrv}
%\bibliography{biblio,biblio2}
\bibliography{Bib_SparseDP}
}

\newpage

\begin{appendix}

\input{appendix}

\end{appendix}

\end{document}

%% file: intro_journal.tex
\paragraph{\textbf{\emph{Motivation and previous work.}}} Treewidth is a fundamental graph parameter that, loosely speaking, measures the resemblance of a graph to a tree. It was introduced by Robertson and Seymour in the early stages of their monumental Graph Minors project~\cite{RoSe86}, but its algorithmic importance originated mainly in Courcelle's theorem~\cite{Cou90}, stating that any graph problem that can be expressed in CMSO logic can be solved in time $f(\tw)\cdot n$ on graphs with $n$ vertices and treewidth $\tw$. Nevertheless, the function $f(\tw)$ given by Courcelle's theorem is unavoidably huge~\cite{FG06}, so from an algorithmic point of view it is crucial to identify problems for which $f(\tw)$ grows {\sl moderately} fast.

Many problems can be solved in time $2^{O(\tw \log \tw)} \cdot n^{O(1)}$ when the $n$-vertex input (general) graph comes equipped with a tree-decomposition of width $\tw$. Intuitively, this is the case of problems that can be solved via dynamic programming on a tree-decomposition by enumerating all {\sl partitions} or {\sl packings} of the vertices in the bags of the tree-decomposition, which are $\tw^{O(\tw)}=2^{O(\tw \log \tw)}$ many. In this article we only consider this type of problems and, more precisely, we are interested in which of these problems can be solved in time $2^{O(\tw)} \cdot n^{O(1)}$; such a running time is called \emph{single-exponential}. This topic has been object of extensive study during the last decade. Let us briefly overview the main results on this line of research.

It is well known that problems that have \emph{locally checkable certificates}\footnote{That is, certificates consisting of a constant number of bits per vertex that can be checked by a cardinality check and by iteratively looking at the neighborhoods of the input graph.}, like \textsc{Vertex Cover} or \textsc{Dominating Set}, can be solved in single-exponential time on general graphs. Intuitively, for this problems it is enough to enumerate {\sl subsets} of the bags of a tree-decomposition (rather than partitions or packings), which are $2^{O(\tw)}$ many. A natural class of problems that do {\sl not} have locally checkable certificates is the class of so-called \emph{connectivity problems}, which contains for example  \textsc{Hamiltonian Cycle}, \textsc{Steiner Tree}, or \textsc{Connected Vertex Cover}. These problems have the property that the solutions should satisfy a {\sl connectivity} requirement (see~\cite{CNP11,BCKN12,RST11} for more details), and using classical dynamic programming techniques it seems that for solving such a problem it is necessary to enumerate partitions or packings of the bags of a tree-decomposition.

A series of articles provided single-exponential algorithms for connectivity problems when the input graphs are restricted to be sparse, namely planar~\cite{DPBF10}, of bounded genus~\cite{RST11,DFT06}, or excluding a fixed graph as a minor~\cite{RST12,DFT12}. The common key idea of these works is to use special types of branch-decompositions (which are objects similar to tree-decompositions) with nice combinatorial properties, which strongly rely on the fact that the input graphs are sparse.

Until very recently, it was a common belief that all problems solvable in single-exponential time of general graphs should have locally checkable certificates, specially after Lokshtanov \emph{et al}.~\cite{LMS11a} proved that one connectivity problem, namely \textsc{Disjoint Paths}, cannot be solved in time $2^{o(\tw \log \tw)} \cdot n^{O(1)}$ on general graphs unless the Exponential Time Hypothesis (ETH) fails\footnote{The ETH states that {\sc 3-SAT} cannot be solved in subexponential time.}. This credence was disproved by Cygan \emph{et al}.~\cite{CNP11}, who provided single-exponential {\sl randomized} algorithms on general graphs for several connectivity problems, like \textsc{Longest Path},  \textsc{Feedback Vertex Set}, or \textsc{Connected Vertex Cover}. More recently, Bodlaender \emph{et al}.~\cite{BCKN12} presented single-exponential {\sl deterministic} algorithms for basically the same connectivity problems, and an alternative proof based on matroids was given by Fomin \emph{et al}.~\cite{FLS13}. These results have been considered a breakthrough, and in particular they imply that essentially {\sl all} connectivity problems that were known to be solvable in single-exponential time on sparse graph classes~\cite{DFT06,RST11,RST12,DFT12,DPBF10} are also solvable in single-exponential time on general graphs~\cite{CNP11,BCKN12}.

\paragraph{\textbf{\emph{Our main results.}}} In view of the above discussion, a natural conclusion is that sparsity may not be particularly helpful or relevant for obtaining single-exponential algorithms. However, in this article we convey that sparsity (in particular, planarity) {\sl does} play a role in connectivity problems parameterized by treewidth. To this end, among the problems that can be solved in time $2^{O(\tw \log \tw)} \cdot n^{O(1)}$ on general graphs, we distinguish the following three disjoint types:
\begin{itemize}
\item[$\bullet$] \textbf{Type 1}: Problems that can be solved in time $2^{O(\tw)} \cdot n^{O(1)}$ on general graphs.
\item[$\bullet$] \textbf{Type 2:} Problems that {\sl cannot} be solved in time $2^{o(\tw \log \tw)} \cdot n^{O(1)}$ on general graphs unless the ETH fails, but that can be solved in time $2^{O(\tw)} \cdot n^{O(1)}$ when restricted to planar graphs.
\item[$\bullet$] \textbf{Type 3:} Problems that {\sl cannot} be solved in time $2^{o(\tw \log \tw)} \cdot n^{O(1)}$ even when restricted to planar graphs, unless the ETH fails.
\end{itemize}

Problems that have locally checkable certificates are of Type~1. As discussed in Section~\ref{sec:type2}, known results imply that there exist problems of Type~2, such as \textsc{Cycle Packing}. Our main contribution is to show that there exist problems of Type~3, thus demonstrating that some connectivity problems can indeed be distinguished according to their behavior on planar graphs. More precisely, we prove the following results:
\begin{itemize}
%\item[$\bullet$] It is known that for some problems a single-exponential running time is best possible unless the ETH fails~\cite{IPZ01}. Nevertheless,
%    %to the best of our knowledge
%    such a result requires an ad-hoc proof for each problem (cf. for instance~\cite{IPZ01,VSX05}). We prove in Section~\ref{sec:type1} that \textsc{3-Colorability}, which is a problem of Type~1, cannot be solved in time $2^{o(\tw)} \cdot n^{O(1)}$  unless the ETH fails, even when the input is a planar graph of maximum degree at most 5.
\item[$\bullet$] In Section~\ref{sec:type2} we provide some examples of problems of Type 2.
%show that \textsc{Cycle Packing} and some other problems are of Type 2 (the lower bound had already %been proved in~\cite{CNP11}).
Furthermore, we prove that \textsc{Planar Cycle Packing} cannot be solved in time $2^{o(\tw)} \cdot n^{O(1)}$ unless the ETH fails, and therefore the running time $2^{O(\tw)} \cdot n^{O(1)}$ is tight.
\item[$\bullet$] In Section~\ref{sec:type3} we provide an example of problem of Type~3: \textsc{Monochromatic Disjoint Paths}, which is a variant of the \textsc{Disjoint Paths} problem on a vertex-colored graph with additional restrictions on the allowed colors for each path. To the best of our knowledge, problems of this type had not been identified before.
\end{itemize}

In order to obtain our results, for the upper bounds we strongly follow the algorithmic techniques based on \emph{Catalan structures} used in~\cite{DFT06,RST11,RST12,DFT12,DPBF10}, and for some of the lower bounds we use the framework introduced in~\cite{LMS11a}, and that has been also used in~\cite{CNP11}.

%\com{Due to space limitations, the proofs of the results marked with `$[\star]$' have been moved to the appendix.}

\paragraph{\textbf{\emph{Additional results and further research.}}} We feel that our results about the role of planarity in connectivity problems parameterized by treewidth are just a first step in a subject that can be much exploited, and we think that the following avenues are particularly interesting:

\begin{itemize}
\item[$\bullet$] It is known that \textsc{Disjoint Paths} can be solved in time $2^{O(\tw \log \tw)} \cdot n^{O(1)}$ on general graphs~\cite{Sch94}, and that this bound is asymptotically tight under the ETH~\cite{LMS11a}. The fact whether \textsc{Disjoint Paths} belongs to Type~2 or Type~3 (or maybe even to some other type in between) remains an important open problem that we have been unable to solve. Towards a possible answer to this question, we prove in Section~\ref{sec:pdp} that \textsc{Planar Disjoint Paths} cannot be solved in time $2^{o(\tw)} \cdot n^{O(1)}$ unless the ETH fails.

%\item[$\bullet$] Another fundamental problem is \textsc{Subgraph Isomorphism}, which is known to be solvable in time $2^{O(h)}\cdot n^{O(1)}$ on planar graphs~\cite{Dor10} and graphs on surfaces~\cite{Bon12}, where $h$ is the number of vertices of a pattern graph $H$ to be found in a host graph $G$ on $n$ vertices. We prove in Section~\ref{sec:otwlogtw} that \textsc{Planar Subgraph Isomorphism} cannot be solved in time $2^{o(\tw \log \tw)}\cdot n^{O(1)}$ unless the ETH fails, but an algorithm running in time  $2^{O(\tw \log \tw)}\cdot n^{O(1)}$ (that is, with no dependency on $H$)  is not known to exist.

\item[$\bullet$] Lokshtanov \emph{et al}.~\cite{LMS11b} have proved that for a number of problems such as \textsc{Dominating Set} or \textsc{$q$-Coloring}, the best known constant $c$ in algorithms of the form $c^{ \tw}\cdot n^{O(1)}$ on general graphs is best possible unless the Strong ETH fails. Is it possible to provide better constants for these problems on planar graphs? The existence of such algorithms would permit to further refine the problems belonging to Type~1.

\item[$\bullet$] Are there NP-hard problems solvable in time $2^{o(\tw)}\cdot n^{O(1)}$?

\item[$\bullet$] Finally, it would be interesting to obtain similar results for problems parameterized by pathwidth, and to extend our algorithms to more general classes of sparse graphs.
%\item[$\bullet$] Our algorithms can be \red{(tediously)} extended to bounded-genus graphs of bounded genus~\cite{RST11} and, more generally, to $H$-minor-free graphs~\cite{RST12,DFT12}.
\end{itemize}

%% file: preliminaries.tex
\paragraph{\textbf{\emph{Graphs}}.}  We use standard graph-theoretic notation, and  the reader is referred to~\cite{Die05}  for any undefined term.  All the graphs we consider are undirected and contain neither loops nor multiple edges. We denote by $V(G)$ the set of vertices of a graph $G$ and by $E(G)$ its set of edges.
%When the context is clear, $n$ stands for the number of vertices of a graph $G$, which will typically be the input graph of the problem under consideration.
A \emph{subgraph} $H = (V_H,E_H)$ of a graph $G=(V,E)$ is a graph such that $V_H \subseteq V$ and $E_H \subseteq E \cap (V_H \times V_H)$.
The \emph{degree} of a vertex $v$ in a graph $G$, denoted by $\degs{G}{v}$, is the number of edges of $G$ containing $v$. The \emph{grid} $m*k$ is the graph $Gr_{m,k} = (\{a_{i,j}| i \in [m], j \in [k]\}, \{(a_{i,j},a_{i+1,j}) | i \in [m-1], j \in [k] \} \cup \{(a_{i,j},a_{i,j+1}) | i \in [m], j \in [k-1] \} )$. When $m=k$ we just speak about the \emph{grid of size $k$}. We say that there is a \emph{path} $s \dots t$ in a graph $G$ if there exist $m \in \N$ and $x_0, \dots, x_m$ in $V(G)$ such that $x_0 = s$, $x_m = t$, and for all $i \in [m]$, $(x_{i-1},x_i) \in E(G)$. 

Throughout the paper, when the problem under consideration is clear, we let $n$ denote the number of vertices of the input graph, $\tw$ its treewidth, and $\pw$ its pathwidth, to be defined below. We use the notation $[k]$ for the set of integers $\{1, \dots, k\}$. In the set $[k]\times [k]$, a \emph{row} is a set $\{i\} \times [k]$ and a \emph{column} is a set $[k] \times \{i\}$ for some $i\in [k]$. If $\mathbf{P}$ is a problem defined on graphs, we denote by {\sc Planar $\mathbf{P}$} the restriction of $\mathbf{P}$ to planar input graphs.

%A \emph{tree} $T$ is a graph with no cycles.
%In a tree, a \emph{leaf} is a vertex with degree $1$.

\paragraph{\textbf{\emph{Treewidth and pathwidth}}.}
A \emph{tree-decomposition} of width $w$ of a graph $G=(V,E)$ is a pair $(T,\sigma)$, where $T$ is a tree and $\sigma = \{ B_t | B_t \subseteq V, t \in V(T) \}$ such that:
\begin{itemize}
\item[$\bullet$] $\bigcup_{t \in V(T)} B_t = V$;
\item[$\bullet$] For every edge $\{u,v\} \in E$ there is a $t \in V(T)$ such that $\{u, v\} \subseteq B_t$;
\item[$\bullet$] $B_i \cap B_k \subseteq B_j$ for all $\{i,j,k\} \subseteq V(T)$ such that $j$ lies on the path $i \dots k$ in $T$;
\item[$\bullet$] $\max_{i \in V(T)} |B_t| = w +1$.
\end{itemize}

The sets $B_t$ are called \emph{bags}. The \emph{treewidth} of $G$, denoted by $\tw(G)$, is the smallest integer $w$ such that there is a tree-decomposition of $G$ of width $w$. An \emph{optimal tree-decomposition} is a tree-decomposition of width $\tw(G)$. A \emph{path-decomposition} of a graph $G = (V,E)$ is a tree-decomposition $(T,\sigma)$ such that $T$ is a path.
The \emph{pathwidth} of $G$, denoted by $\pw(G)$, is the smallest integer $w$ such that there is a path-decomposition of $G$ of width $w$. Clearly, for any graph $G$, we have $\tw(G) \leq \pw (G)$.

\paragraph{\textbf{\emph{Branchwidth}}.}
A \emph{branch-decomposition} $(T, \sigma)$ of a graph $G=(V,E)$ consists of an unrooted ternary tree $T$ and a bijection $\sigma : L \rightarrow E$ from the set $L$ of leaves of $T$ to the edge set of $G$. We define for every edge $e$ of $T$ the \emph{middle set} $\mids (e) \subseteq V(G)$ as follows: Let $T_1$ and $T_2$ be the two connected components of $T \bs \{e\}$. Then let $G_i$ be the graph induced by the edge set $\{\sigma (f) : f \in L \cap V(T_i)\}$ for $i \in \{1,2\}$. The \emph{middle set} of $e$ is the intersection of the vertex sets of $G_1$ and $G_2$, i.e., $\mids (e) := V(G_1) \cap V(G_2)$. When we consider $T$ as rooted, we let $G_e$ be
 the graph $G_i$ such that $T_i$ does not contain the root of $T$.
The \emph{width} of $(T,\sigma)$ is the maximum order of the middle sets over all edges of $T$, i.e., $w(T,\sigma) := \max\{|\mids (e) | | e \in T\}$.
The \emph{branchwidth} of $G$, denoted by $\bw(G)$, is the minimum width over all branch decompositions of $G$.
An \emph{optimal branch decomposition} of $G$ is a branch decomposition $(T,\sigma)$ of width $\bw(G)$. By~\cite{Rob91}, the branchwidth of a graph $G$ with at least 3 edges is related to its treewidth by $\bw(G) -1 \leq \tw(G) \leq \lfloor\frac{3}{2} \bw(G) \rfloor -1$.

\paragraph{\textbf{\emph{Planar graphs}}.}
Let $\Sigma$ be the sphere $\{(x,y,z) \in \mathbb{R}^3 : x^2 + y^2 + z^2 = 1\}$. By a \emph{$\Sigma$-plane} graph $G$ we mean a planar graph $G$ with its vertex set $V(G)$, edge set $E(G)$, and face set $F(G)$ drawn without edge crossings in $\Sigma$. An $O$-arc is a subset of $\Sigma$ homeomorphic to a circle. An $O$-arc in $\Sigma$ is called a \emph{noose} of a $\Sigma$-plane graph $G$ if it meets $G$ only in vertices and intersects with every face at most once. Each noose $O$ bounds two open discs $\Delta_1, \Delta_2$ in $\Sigma$, i.e., $\Delta_1 \cap \Delta_2 = \es$ and $\Delta_1 \cup \Delta_2 \cup O = \Sigma$.

For a $\Sigma$-plane graph $G$, we define a \emph{sphere cut decomposition} $(T, \sigma, \pi)$ of $G$, or \emph{sc-decomposition} for short, as a branch-decomposition such that for every edge $e$ of $T$ there exists a noose $O_e$ bounding the two open discs $\Delta_1$ and $\Delta_2$ such that $G_i \subseteq \Delta_i \cup O_e$, $1 \leq i \leq 2$. Thus $O_e$ meets $G$ only in $\mids (e)$ and its length is $|\mids (e)|$. It is known that any planar graph $G$ has a sc-decomposition of width $\bw(G)$ that can be computed in polynomial time~\cite{DPBF10,SeTh94}.

\paragraph{\textbf{\emph{Non-crossing partitions and matchings}}.}
A \emph{partition} $P$ of a set $S$ is a set of subsets of $S$ such that $\bigcup_{s \in P} s = S$ and for all distinct $s_1, s_2 \in P$, $s_1 \cap s_2 = \es$.
A partition $P$ is called \emph{non-crossing partition} if for each $s_1, s_2 \in P$, for each $a, b \in s_1$ and $c,d \in s_2$ with $a < b$ and $c < d$ then one of the following situations occurs:
$a < b < c < d$, $a < c < d < b$, $c < d < a < b$, or $c < a < b < d$.
Kreweras showed in \cite{Kre72} that the number of non-crossing partitions on $[k]$ for $k \in \N$ is at most $4^k$.

%\paragraph{Matchings.}

A \emph{matching} $M$ is a set of pairs of elements of a set, which we also call edges, such that for each $e,e' \in M$, $e \not = e'$, $e \cap e' = \es$.
For a matching $M$ in a graph $G$, we denote by $\set{M}$ the set of all vertices that belong to an edge of $M$.
We say that two matchings $M$ and $M'$ are \emph{disjoint} if $\set{M} \cap \set{M'} = \es$. A matching $M$ on $V=\{v_1, \dots, v_n\}$, for some $n \in \N^*$, is called \emph{non-crossing matching} if for each $\{v_a,v_b\}, \{v_c,v_d\} \in M$, with $a < b$ and $c < d$, then one of the following situations occurs:
$a < b < c < d$, $a < c < d < b$, $c < d < a < b$, or $c < a < b < d$. Kreweras showed in~\cite{Kre72} that the number of non-crossing matchings on $[k]$ for $k \in \N$ is at most $2^k$.

\paragraph{\textbf{\emph{Tight problems of Type~1}}.} It is usually believed that NP-hard problems parameterized by $\tw$ cannot be solved in time $2^{o(\tw)} \cdot n^{O(1)}$ under some reasonable complexity assumption. This has been proved in~\cite{IPZ01} for problems on {\sl general} graphs such as \textsc{$q$-Colorability, Independent Set,} or {\sc Vertex Cover}, or in~\cite{VSX05} for \textsc{Planar Hamiltonian Cycle}, all these results assuming the ETH.
 %In these reductions the main point is that the relevant parameters should not increase more than linearly.
Nevertheless, such a result requires an ad-hoc proof for each problem. For instance, to the best of our knowledge, such lower bounds are not known for {\sc Cycle Packing} when the input graph is restricted to be {\sl planar}.  In order to deal with {\sc Planar Cycle Packing} in Section~\ref{sec:type2}, we need to introduce the \textsc{3-Colorability} problem, which is easily seen to be a problem of Type~1.

%By making reduction that preserved subexponential complexity, the lower bound $2^{o(\tw)}\cdot n^{O(1)}$ was proved in
%For these reductions the main point is that the relevant parameters should not increase more than linearly.
%To our best knowledge, there is no such reduction for {\sc Planar 3-Colorability} or {\sc Planar Cycle Packing}.
%We show that there is a reduction where the parameter increases quadratically, which turns out to be enough for proving an exponential lower bound in the treewidth.

\probl
{{\sc 3-Colorability}}
{An $n$-vertex graph $G=(V,E)$.}
{Is there a coloring $c : V \to \{1,2,3\}$ s.t. for all $\{x,y\} \in E$, $c(x) \not = c(y)$?}

\vspace{-.25cm}

It is known that \textsc{Planar 3-Colorability} cannot be solved in time $2^{o(\sqrt{n})} \cdot n^{O(1)}$ even when the input graph has maximum degree at most 4, unless the ETH fails. Indeed, the NP-completeness reduction for \textsc{Planar 3-Colorability} given in~\cite{GJS76} starts with a {\sc 3-SAT} formula on $t$ variables, and creates a planar graph of maximum degree at most 4 in which the number of vertices is $O(t^2)$, yielding directly the desired result. %The reduction of~\cite{GJS76} consists of a sequence of reductions between intermediate problems. 
In Theorem~\ref{th:lb3c} below, whose proof can be found in Appendix~\ref{sec:3coloring}, we provide an alternative proof of this result with a slightly worse degree bound.

\begin{theorem}
\label{th:lb3c}%$[\star]$
\textsc{Planar 3-Colorability} cannot be solved in time $2^{o(\sqrt{n})} \cdot n^{O(1)}$  unless the ETH fails, even when the input graph has maximum degree at most 5.
\end{theorem}
%\begin{proof}
%\end{proof}

\begin{corollary}
\textsc{Planar 3-Colorability} cannot be solved in time $2^{o(\tw)} \cdot n^{O(1)}$  unless the ETH fails,  even if the input graph has maximum degree at most 5.
\end{corollary}

\begin{proof}
As a planar graph $G$ on $n$ vertices satisfies $\tw(G) = O(\sqrt{n})$ \cite{Fom06},  an algorithm in time $2^{o(\tw)}\cdot n^{O(1)}$ for {\sc Planar 3-Colorability} implies that there is an algorithm in time $2^{o(\sqrt{n})} \cdot n^{O(1)}$, which is impossible by Theorem~\ref{th:lb3c} unless the ETH fails.
\end{proof}

%% file: type2.tex
In this section we deal with problems of Type~2. Let us start with the following problem.

%In this section we prove that the \textsc{Cycle Packing} problem is of Type~2. Other problems of Type~2 are discussed in Appendix~\ref{sec:otherType2}.

%More precisely, we prove that is of Type~2.

%\subsection{Algorithm for Planar Cycle Packing}
\probls
{Cycle Packing}
{An $n$-vertex graph $G=(V,E)$ and an integer $\ell_0$.}
{Does $G$ contain $\ell_0$ pairwise vertex-disjoint cycles?}
{The treewidth $\tw$ of $G$.}

It is proved in~\cite{CNP11} that \textsc{Cycle Packing} {\sl cannot} be solved in time $2^{o(\tw \log \tw)} \cdot n^{O(1)}$ on general graphs unless the ETH fails. On the other hand, a dynamic programming algorithm for \textsc{Planar Cycle Packing} running in time $2^{O(\tw)} \cdot n^{O(1)}$ can be found in~\cite{KLL02}. Therefore, it follows that \textsc{Cycle Packing} is of Type~2. In Lemma~\ref{lem:CyclePacking} below we provide an alternative algorithm for \textsc{Planar Cycle Packing} running in time $2^{O(\tw)} \cdot n^{O(1)}$, which is a direct application of the techniques based on \emph{Catalan structures} introduced in~\cite{DPBF10}. We include its proof here for completeness, as it yields slightly better constants than the algorithm of~\cite{KLL02}, and because we will use similar terminology in the more involved algorithm of Lemma~\ref{lem:algoMonochromDisjointPaths} in Section~\ref{sec:type3}.

\begin{lemma}\label{lem:CyclePacking}%$[\star]$
\textsc{Planar Cycle Packing} can be solved in time $2^{O(\tw)} \cdot n^{O(1)}$.
\end{lemma}
\begin{proof} We prove the lemma for $\bw$, but as  $\tw \leq \lfloor\frac{3}{2} \bw \rfloor -1$, it will imply the same asymptotic upper bound for $\tw$. Let $G$ be a graph, $X \subseteq V(G)$, and $M$ a  matching on  $V(G)\bs X$.
Intuitively, $M$ represents the endpoints of the paths we are building and $X$ is the set of vertices that are already inside a path but they are not an endpoint of any path.
We define $G[(X,M,\ell)] = (\set{M},M)$.
We say that $G[(X_1, M_1,\ell_1), (X_2,M_2,\ell_2)]$ is \emph{defined} if $X_1 \cap (X_2 \cup \set{M_2}) = X_2 \cap (X_1 \cup \set{M_1}) = \es$ and we define $G[(X_1, M_1,\ell_1), (X_2,M_2,\ell_2)] = G[(X_1,M_1,\ell_1)] \cup G[(X_2, M_2,\ell_2)]$. Otherwise, we say that $G[(X_1, M_1,\ell_1), (X_2,M_2, \ell_2)]$ is \emph{undefined}.
We say that $cp(G,X,M) \geq \ell$ if $G$ contains paths joining each pair of vertices given by  $M$ and $\ell$ cycles, all pairwise vertex-disjoint.

We now consider $G=(V,E)$ to be our $\Sigma$-plane input graph and $\ell_0$ our integer. Let $(T,\mu,\pi)$ be a sc-decomposition of $G$ of width $\bw$. As in \cite{DPBF10}, we root $T$ by arbitrarily choosing an edge $e$ and we subdivide it by inserting a new node $s$. Let $e'$ and $e''$ be the new edges and set $\mids(e') = \mids(e'')=\mids(e)$. We create a new node root $r$, we connect it to $s$ by an edge $e_r$, and set $\mids(e_r) = \emptyset$. The root $e_r$ is not considered as a leaf.

Let $e \in E(T)$ and $\RR_e = \{(X,M,\ell) | X \subseteq \mids (e)$, $M$ is a  matching of a subset of $\mids (e) \bs X$, and $cp(G_e,X,M) \geq \ell\}$. We observe that there exist $\ell_0$ pairwise vertex-disjoint cycles in $G$ if and only if $(\emptyset, \emptyset, \ell_0) \in \RR_{e_r}$.
We should now compute $\RR_{e_r}$. If $e$ is a leaf then $G_e = (\{x,y\}, \{(x,y)\})$ and $\RR_e = \{(\es,\es,0), (\es,\{(x,y)\},0)\}$.
Otherwise, let $e_1$ and $e_2$ be the two children of $e$ in $E(T)$. $\RR_e$ is the set of all triples $(X,M,\ell)$ such that there exist $(S_1,S_2) = ((X_1,M_1,\ell_1),(X_2,M_2,\ell_2)) \in \RR_{e_1} \times \RR_{e_2}$ such that
$M \subseteq ((\set{M_1}\cup \set{M_2})\cap (\mids (e)\bs X))^2$, $G[S_1,S_2]$ is defined, all vertices in $\mids (e)$ of degree at least two in $G[S_1,S_2]$ are in X, and we can find in $G[S_1,S_2]$ $\ell_3$ cycles and a path $x \dots y$ for each $(x,y) \in M$
such that $\min(\ell_1 + \ell_2 + \ell_3, \ell_0) \geq \ell$.

Note that $G[S_1,S_2]$ is a minor of $G$ so $G[S_1,S_2]$ is also planar. As we have considered a sc-decomposition and  all the paths we consider in $G[S_1,S_2]$ are pairwise vertex-disjoint, since each vertex has degree at most two, the maximum number of distinct  matchings $M$ is bounded by the number of non-crossing matchings on $|\mids(e)|$ elements, which is at most $2^{|\mids(e)|}$. As we have at most $3^{|\mids(e)|}$ choices for $X$ and $\set{M}$, it follows that for each $e \in E(T)$, $|\RR_e| \leq 6^{|\mids(e)|}\cdot \ell_0$.
As for each $e \in E(T)$ such that $e$ is not a leaf, we have to merge the tables of the two children $e_1$ and $e_2$ of $e$, this algorithm can check in time $O(36^{\bw}\cdot \ell_0^2 \cdot |V(G)|)$ whether $G$ contains at least $\ell_0$ vertex-disjoint cycles. We note that  the constant can probably be optimized, for example by using fast matrix multiplication (see for instance~\cite{Wil12}),  but this is outside of the scope of this paper.\end{proof}

%\subsection{Lower bound for Planar Cycle Packing}

\vspace{.25cm}

We now prove that the running time given by Lemma~\ref{lem:CyclePacking} is asymptotically tight.

\begin{theorem}
\label{th:lbcp}
\textsc{Planar Cycle Packing} cannot be solved in time $2^{o(\sqrt{n})} \cdot n^{O(1)}$ unless the ETH fails. Therefore, \textsc{Planar Cycle Packing} cannot be solved in time $2^{o(\tw)} \cdot n^{O(1)}$ unless the ETH fails.
\end{theorem}

\begin{proof}
To prove this theorem, we reduce from \textsc{Planar 3-Colorability} where the input graph has maximum degree at most 5. Let $G=(V,E)$ be a planar graph with maximum degree at most 5 with $V = \{v_1, \dots, v_n\}$. We proceed to construct a planar graph $H$ together with a planar embedding of it, where we will ask for an appropriate number $\ell_0$ of vertex-disjoint cycles.

In this proof, we abuse notation and say that we \emph{ask} for $x$ cycles in a gadget to say that the number of cycles we are looking for in the {\sc Planar Cycle Packing} problem is increased by $x$.
We will ask for a certain number of cycles in each of the introduced gadgets, which by construction will lead to a set of cycles of maximum cardinality in $H$.

We start by introducing some gadgets.
For each $i \in [n]$, corresponding to the vertices $v_1, \dots, v_n$ of $G$, we add to $H$ the SC$_i$-gadget depicted in Fig. \ref{fig:CPSCG}. More precisely, $SC_i = (\{a_i,b_i,c_i, u_{i,0}, u_{i,1}, u_{i,2}, u_{i,3}\},\{(u_{i,0},u_{i,1}),(u_{i,0},u_{i,2}),(u_{i,0},u_{i,3}), (a_i,u_{i,1}),(a_i,u_{i,2}),$ $(b_i,u_{i,2}), (b_i,u_{i,3}), (c_i,u_{i,1}), (c_i,u_{i,3}) \})$. We ask for a cycle inside this gadget. This cycle imposes that at least one of the vertices $\{a_i, b_i, c_i\}$, named a \emph{selected vertex} of the SC$_i$-gadget, is used by the inner cycle and leaves the possibility that the two others are free.
The intended meaning of each SC$_i$-gadget is as follows. The three vertices $a_i$, $b_i$, and $c_i$ correspond to the three colors in the 3-coloring of $G$, namely $a$, $b$, and $c$. If for instance $a_i$ is a selected vertex for index $i$, it will imply that vertex $v_i$ can be colored with color $a$. Therefore, each SC$_i$-gadget defines the available colors for vertex $v_i$, which we call the \emph{color output} of vertex $v_i$.

\input{cycle_packing_figure}

In order to construct a graph $H$ that defines a valid 3-coloring of $G$, we need to propagate the color output of $v_i$ as many times as the degree of $v_i$ in $G$. For this, we introduce a gadget called \emph{bifurcate} gadget.
Before proceeding to the description of the gadget, let us describe its intended functionality. The objective is, starting with the vertices $a_i$, $b_i$, and $c_i$ of the SC$_i$-gadget, to construct a set of triples $\{a_{i,k}, b_{i,k}, c_{i,k}\}$ for $1 \leq k \leq \degs{G}{v_i}$ such that in each triple there will be again at least one selected vertex, defined by the cycles that we will construct in the bifurcate gadgets. Note that in the SC$_i$-gadget the choice of a selected vertex in each triple $\{a_{i,k}, b_{i,k}, c_{i,k}\}$ naturally defines a color output for vertex $v_i$. The crucial property of the gadget is that the intersection of the color outputs given by all the triples is non-empty if and only if the graph $H$ contains enough vertex-disjoint cycles. In other words, the existence of the appropriate number of vertex-disjoint cycles in $H$ will define an available color for each vertex $v_i$ of $G$.

%\ig{put some parts of this proof in the appendix...}

We now proceed to the construction of the bifurcate gadget. First we need to introduce three other auxiliary gadgets. The first two ones, called \emph{expel} and \emph{double-expel} gadgets, are depicted in Fig. \ref{fig:CPE}. Formally, for two vertices $u$ and  $u'$, the expel gadget is defined as  $EG_{u, u'} =(\{u,u',v,v'\}, \{(u,v), (u,v'), (u',v), (u',v'), (v,v')\})$, and we ask for a cycle inside each such expel gadget. This gadget ensures that if $u$ is in another cycle, then $u'$ is necessarily used by the internal cycle and vice-versa.
Similarly, the double-expel gadget for three vertices $u$, $u'$, and $u''$ is defined as $DEG_{u, u',u''} =(\{u,u', u'', v, v'\}, \{(u,v), (u,v'), (u',v), (u'',v'), (u',u''), (v,v')\})$, and we also ask for a cycle inside each such gadget. This gadget ensures that if $u$ is in another cycle, then $u'$ and $u''$ are necessarily used by the internal cycle and that if $u'$ or $u''$ are in an external cycle, then $u$ is necessarily used by the internal cycle.

\input{bifurcate}

As in our construction the edges of the expel gadgets will cross, we need a gadget that replaces each edge-crossing with a planar subgraph while preserving the existence of the original edges, in the sense that each of the crossing edges gets replaced by a path joining the endvertices of the original edge. This gadget is called \emph{path-crossing} gadget and is depicted in Fig. \ref{fig:CPPC}. Formally, the path-crossing gadget $PCG$ is such that
$ \{pc_{1}, pc_{2}, pc_{3}, pc_{4}, w_0, w_{ 1,1}, w_{ 1,2}, w_{ 2,1}, w_{ 2,2}, w_{ 3,1}, w_{ 3,2}, w_{ 4,1}, w_{ 4,2}\} \subseteq V(PCG)$, $E(PCG)$ contains two paths $pc_{1}, w_{ 1,1}, w_{ 1,2}, w_{0}, w_{ 3,2}, w_{ 3,1}, pc_{3}$ and $pc_{2}, w_{ 2,1}, w_{ 2,2}, w_{0}, w_{ 4,2}, w_{ 4,1}, pc_{4}$, and we add 4 expel gadgets $EG_{w_{1,1},w_{2,2}}$,  $EG_{w_{2,1},w_{3,2}}$,  $EG_{w_{3,1},w_{4,2}}$,  $EG_{w_{4,1},w_{1,2}}$ to $PCG$.
 We ask in this gadget only the 4 cycles asked in the expel gadgets. This gadget ensures that, in order to have enough vertex-disjoint cycles, an external cycle that contains an edge from a path-crossing gadget should go \emph{straight}, i.e., for all $\alpha \in [4]$, if the cycle arrives at a vertex $pc_{\alpha}$ it should leave by $pc_{(\alpha+1 \pmod{4})+1}$. If a cycle does not respect this property, we say that the cycle \emph{turns} inside the path-crossing gadget.
That is, the gadget preserves the existence of the original crossing edges whenever there are no cycles that turn inside it. Note that the two paths corresponding to the two original crossing edges cannot be used simultaneously by a set of cycles in the planar graph $H$. We can now define the bifurcate gadget, which is depicted in Fig.~\ref{fig:CPEC}(a), and where each of the 12 edge-crossings should be replaced by a path-crossing gadget. Note that each bifurcate gadgets contains 6 expel and 3 double-expel gadgets. We ask in this gadget the 48 cycles of the path-crossing gadgets, the 3 cycles of the double expel gadgets, and the 6 cycles of the expel gadgets. Note that, indeed, given a triple $\{a_i, b_i, c_i\}$ defining a color output for a vertex $v_i$, the cycles asked in the bifurcate gadget define two triples $\{a_{i,1}, b_{i,1}, c_{i,1}\}$ and  $\{a_{i,2}, b_{i,2}, c_{i,2}\}$, which in turn define two color outputs compatible with the one defined by $\{a_i, b_i, c_i\}$, in the sense that there is a common available color for $v_i$. For example, in Fig.~\ref{fig:CPEC}(b) vertex $a_i$ is the only selected vertex of $\{a_i, b_i, c_i\}$ (given by the corresponding SC$_i$-gadget, which is not shown in the figure for the sake of visibility), and the bold cycles define the selected vertices for the triples $\{a_{i,1}, b_{i,1}, c_{i,1}\}$ and  $\{a_{i,2}, b_{i,2}, c_{i,2}\}$. Note that color $a$ is simultaneously available for the three triples. We would like to stress that there are other choices of a maximum-cardinality set of cycles in the bifurcate gadget of Fig.~\ref{fig:CPEC}(b), but all of them yield color $a$ available.
For each vertex $v_i$, we need as many triples $\{a_{i,k}, b_{i,k}, c_{i,k}\}$ as $\degs{G}{v_i}$. For that, we concatenate the bifurcate gadgets $deg_G(v_i) -1$ times in the following way. Inductively, we consider the triple $\{a_{i,2}, b_{i,2}, c_{i,2}\}$ of Fig. \ref{fig:CPEC}(a) as the original triple $\{a_i, b_i, c_i\}$, and plug another bifurcate gadget starting from this triple.

With the gadgets defined so far, we have a representation of the colored vertices of $G$ in $H$. We now proceed to capture the edges of $G$ in $H$. For this, we introduce for each $\{v_i, v_j\} \in E$, $i,j \in [n]$, an \emph{edge} gadget depicted in Fig. \ref{fig:CPEG}, where all the 12 edge-crossings should be replaced by a path-crossing gadget. We ask in this gadget 51 new cycles (3 for the expel gadgets and 48 for the path-crossing gadgets). We plug one side of this gadget to a triple $\{a_{i,k}, b_{i,k}, c_{i,k}\}$ defining a color output of $v_i$ and the other side to a triple $\{a_{j,k'}, b_{j,k'}, c_{j,k'}\}$ defining a color output of $v_j$. The edge gadget ensures that the intersection of the two color outputs is empty.  This completes the construction of $H$, which is clearly a planar graph, and we set $\ell_0$ to be the sum of the number of cycles asked in each of the introduced gadgets.

\input{edge-gadget}

\begin{claimN}\label{claim1}%\emph{[$\star$]}
In any solution of \textsc{Cycle Packing} in $H$, each expel gadget, double-expel gadget, and SC$_i$-gadget contains a cycle, and each cycle is contained inside such a gadget.
\end{claimN}
\begin{proof} In this proof, we say that a cycle $C$ \emph{kills} an another cycle $C'$ if, for any set $S$ of vertex-disjoint cycles containing $C$,  $(S \bs \{C\}) \cup \{C'\}$ is also a set of vertex-disjoint cycles. When dealing with a gadget $F$, we say that a cycle intersecting $F$ is \emph{internal} if it contains only vertices in $F$, and \emph{external} otherwise.

First note that any internal cycle in an expel or a double-expel gadget should use both vertices $v$ and $v'$. Also note that if some external cycle in an expel or a double-expel gadget uses the vertex $v$ or $v'$ of an expel or a double-expel gadget, then it also uses the vertex $u$ (or $u$ and $u''$), and then we are not able to find an internal cycle anymore. Therefore, any external cycle containing $v$ or $v'$ kills the cycle on the set of vertices $\{u, v, v'\}$ or $\{ u', u'', v', v''\}$.

Note that if an external cycle of a path-crossing gadget turns inside it, then without loss of generality it uses a path of the form $pc_1, w_{1,1}, w_{1,2}, w_0, w_{2,2}, w_{2,1}, pc_2$ inside the path-crossing gadget. This external cycle kills the cycle inside the expel gadget between $w_{1,1}$ and $w_{2,2}$. Moreover, note that another disjoint external cycle turning in the same path-crossing gadget kills another internal cycle in the path-crossing gadget, namely the one inside the expel gadget between $w_{3,1}$ and $w_{4,2}$.

Let $C$ be a cycle in $H$ that is not entirely contained in only one expel, double-expel, or SC$_i$-gadget. Because of the previous remarks, we have that $C$ cannot turn in two different path-crossing gadgets, and that if it does {\sl not} turn in any path-crossing gadget, then by construction it uses at least two expel or double-expel gadgets and kills their internal cycles. In both configurations, adding $C$ to the solution decreases the number of vertex-disjoint cycles that we can find in $H$.

The only remaining choice for $C$ is to turn exactly once in one path-crossing gadget. If it happens inside a bifurcate gadget, then $C$ uses vertices of two expel gadgets, namely $expel_1$ and $expel_2$, corresponding to two different colors.
The only way to connect vertices corresponding to different colors outside the path-crossing gadget is by using an SC$_i$-gadget.
So either $C$ kills the cycles of $expel_1$ and $expel_2$, or it may also use a path leading to an edge gadget. If $C$ turns in a path-crossing gadget inside an edge gadget, then the analysis is similar, but there is an extra case where the edge gadget representing the edge between $v_i$ and $v_j$ is directly plugged into the SC$_j$-gadget. In this case, note that none of the vertices $a_i, b_i, c_i$ can be a selected vertex with the set of cycles we currently ask for, and therefore in order to allow it we need to decrease the number of cycles in the solution. \end{proof}

\vspace{.2cm} 

If we are given a solution of \textsc{Planar Cycle Packing} in $H$, then for each $i \in [n]$, by Claim~\ref{claim1} the selection of a cycle in the SC$_i$-gadget selects a color for $v_i$, that can be any color that belongs simultaneously to all color outputs of $v_i$, and the edge gadgets ensure that two adjacent vertices are in two different color classes. So in this way we obtain a solution of  \textsc{Planar 3-Colorability} in $G$.

Conversely, given a solution of  \textsc{Planar  3-Colorability} in $G$, we construct a solution of \textsc{Planar Cycle Packing} in $H$ as follows. For each $i \in [n]$ we choose in the SC$_i$-gadget the cycle of length 4 that contains $u_{i,0}$ and the vertex in $\{a_i, b_i, c_i\}$ that corresponds to the color of $v_i$.
We also choose in the bifurcates gadgets the cycles selecting vertices in $\{a_{i,1}, b_{i,1}, c_{i,1}, a_{i,2}, b_{i,2}, c_{i,2}\}$ that lead to two identical color outputs coinciding with the color output of $\{a_i, b_i, c_i\}$. This choice has the property that the color output of $\{a_i, b_i, v_i\}$ is a subset of the color output of $\{a_{i,1}, b_{i,1}, c_{i,1}\}$ and the color output of $\{a_{i,2}, b_{i,2}, c_{i,2}\}$, and leaves as many free vertices as possible for other cycles in other gadgets.  Inside the edge gadget representing $\{v_i, v_j\} \in E$, we select the three cycles that are allowed by the free vertices. We complete our cycle selection by selecting a cycle in each expel gadget contained in a path-crossing gadget. By Claim~\ref{claim1}, this choice leads to a solution of {\sc Planar Cycle Packing} in $H$.

As the degree of each vertex in $G$ is bounded by 5, the number of gadgets we introduce for each $v_i \in V(G)$ to construct $H$ is also bounded by a constant, so
the total number of vertices of $H$ is linear in the number of vertices of $G$.
Therefore if we could solve {\sc Planar Cycle Packing} in time $2^{o(\sqrt{n})}\cdot n^{O(1)}$ then we could also solve {\sc Planar 3-coloring} in time $2^{o(\sqrt{n})} \cdot n^{O(1)}$, which is impossible by Theorem \ref{th:lb3c} unless the ETH fails.\end{proof}
%The theorem follows.

%From Theorem \ref{th:lbcp} we directly obtain the following corollary.

%\begin{corollary}\label{cor:planarCyclePacking}
%The \textsc{Planar Cycle Packing} problem cannot be solved in time $2^{o(\tw)} \cdot n^{O(1)}$ %unless the ETH fails.
%\end{corollary}

%Lemma~\ref{lem:CyclePacking} and Corollary\label{cor:planarCyclePacking} imply %that \textsc{Planar Cycle Packing} is of

\vspace{.15cm}

\paragraph{\textbf{\emph{Other problems of Type~2.}}} We can provide other examples of problems of Type~2. This is the case, for instance, of \textsc{Cycle Cover}, for which the lower bound has been proved in~\cite{CNP11}, and the upper bound can be proved similarly to Lemma~\ref{lem:CyclePacking}.

Other problems of Type~2 are those where one wants to {\sl maximize} the number of connected components induced by the vertices in a solution. It has been proved in~\cite{CNP11} that \textsc{Maximally Disconnected Dominating Set} cannot be solved in time $2^{o(\tw \log \tw)} \cdot n^{O(1)}$ unless the ETH fails. Again, the upper bound can be proved similarly to Lemma~\ref{lem:CyclePacking}. We can define more problems of this flavor, such as the following one.

\probls
{{\sc Maximally Disconnected Feedback Vertex Set}}
{A graph $G=(V,E)$ and two integers $\ell$ and $r$.}
{Does $G$ contain a feedback vertex set of size at most $\ell$ that induces at least $r$ connected components?}
{The treewidth $\tw$ of $G$.}

The following lemma can be proved by using the reduction given in~\cite{CNP11} for {\sc Maximally Disconnected Dominating Set}, just by appropriately redefining the so-called \emph{force} and \emph{one-in-many} gadgets.

\begin{lemma}
\textsc{Maximally Disconnected Feedback Vertex Set} cannot be solved in time $2^{o(\tw \log \tw)} \cdot n^{O(1)}$ unless the ETH fails.
\end{lemma}

And again, the following lemma can be proved using standard dynamic programming techniques.

\begin{lemma}\label{lem:algoDisconnectedFVS}
\textsc{Maximally Disconnected Feedback Vertex Set} can be solved in time $2^{O(\tw \log \tw)} \cdot n^{O(1)}$, and \textsc{Planar Maximally Disconnected Feedback Vertex Set} can be solved in time $2^{O(\tw)} \cdot n^{O(1)}$.
\end{lemma}

%% file: cycle_packing_figure.tex
% SC_i gadget
\begin{figure}[t!]
  \centering
  \begin{tikzpicture}
    \draw
    (0,0) \nodd{label=left:$u_{i,1}$}
    |- +(3,- 2) node [place, label=right:$c_i$] {}
    -- +(2,- 1) \nodd{label=left:$u_{i,3}$}
    -- +(3,0) node [place, label=right:$b_i$] {}
    -- +(2,1) node [place] {}
    -- +(3,2) node [place, label=right:$a_i$] {}
    -| +(0,0)
    -- +(1, 0) node [place] {}
    -- +(2,1) \nodd{label=left:$u_{i,2}$}
    -- +(1,0) \nodd{label=right:$u_{i,0}$}
    -- +(2,- 1) ;
  \end{tikzpicture}

  \caption{The SC$_i$-gadget.}\vspace{-.25cm}
  \label{fig:CPSCG}
\end{figure}
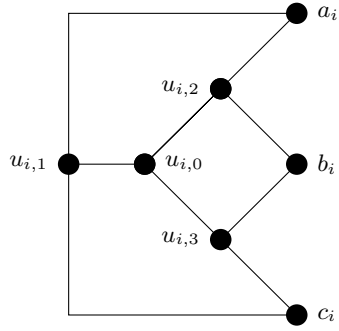

% Expel and double-expel gadget
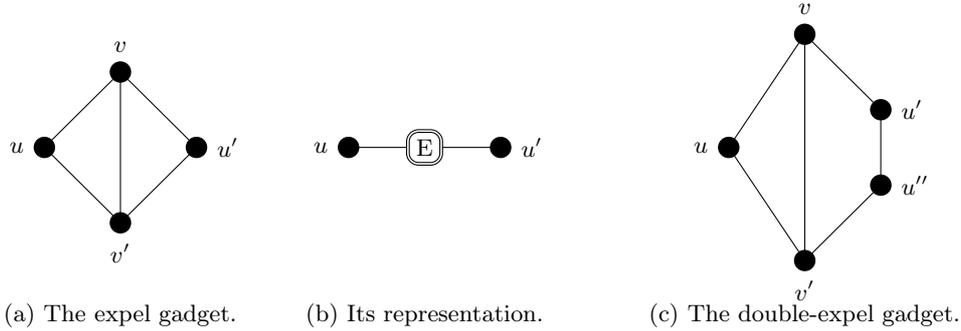
\begin{figure}[t!]
  \centering
  \begin{tikzpicture}
    \draw ( 0, -0.5) \nodd{label=above:$v$}
    -- +  ( 1,-1) \nodd{label=right:$u'$}
    -- +  ( 0,-2) \nodd{label=below:$v'$}
    -- +  (-1,-1) \nodd{label=left:$u$}
    -- +  ( 0, 0)
    -- +  ( 0,-2);
    \capt{-1,-4}{(a) The expel gadget.};
    \draw ( 3, -1.5) \nodd{label=left:$u$}
    -- +  ( 1, 0) \noddd{}{E}
    -- +  ( 2, 0) \nodd{label=right:$u'$};
    \capt{3,-4}{(b) Its representation.};
    \draw ( 9, 0) \nodd{label=above:$v$}
    -- +  ( 1,-1) \nodd{label=right:$u'$}
    -- +  ( 1,-2) \nodd{label=right:$u''$}
    -- +  ( 0,-3) \nodd{label=below:$v'$}
    -- +  (-1,-1.5) \nodd{label=left:$u$}
    -- +  ( 0, 0)
    -- +  ( 0,-3);
    \capt{8,-4}{(c) The double-expel gadget.};
  \end{tikzpicture}

  \caption{The expel gadget and the double-expel gadget.}
  \label{fig:CPE}
\end{figure}

%% Path crossing
\begin{figure}[t!]
  \centering
  \begin{tikzpicture}
    \draw ( 0,-3) \nodd{label=below:$pc_4$}
    --   ( 0, 3) \nodd{label=above:$pc_2$};
    \draw (-3, 0) \nodd{label=left:$pc_1$}
    --    ( 3, 0) \nodd{label=right:$pc_3$};
    \draw ( 0, 0) \nodd{label=70:$w_0$};
    \draw ( 0, 0)
    -- +  (-2, 0) \nodd{label=70:$w_{1,1}$}
    -- +  (-2, 1) \noddd{}{E}
    -- +  ( 0, 1) \nodd{label=70:$w_{2,2}$}
    -- +  ( 0, 2) \nodd{label=70:$w_{2,1}$}
    -- +  ( 1, 2) \noddd{}{E}
    -- +  ( 1, 0) \nodd{label=70:$w_{3,2}$}
    -- +  ( 2, 0) \nodd{label=70:$w_{3,1}$}
    -- +  ( 2,-1) \noddd{}{E}
    -- +  ( 0,-1) \nodd{label=70:$w_{4,2}$}
    -- +  ( 0,-2) \nodd{label=70:$w_{4,1}$}
    -- +  (-1,-2) \noddd{}{E}
    -- +  (-1, 0) \nodd{label=70:$w_{1,2}$};
    \capt{-1,-4}{The gadget.};

    \draw (8, 1) \nodd{label=above:$pc_2$} -- +(0,-2) \nodd{label=below:$pc_4$};
    \draw (7, 0) \nodd{label=left:$pc_1$} -- +(1,0) \noddd{}{PC}
    -- +(2,0) \nodd{label=right:$pc_3$};
    \capt{7,-4}{The representation.};
  \end{tikzpicture}

  \caption{Path-crossing gadget.}
  \label{fig:CPPC}
\end{figure}
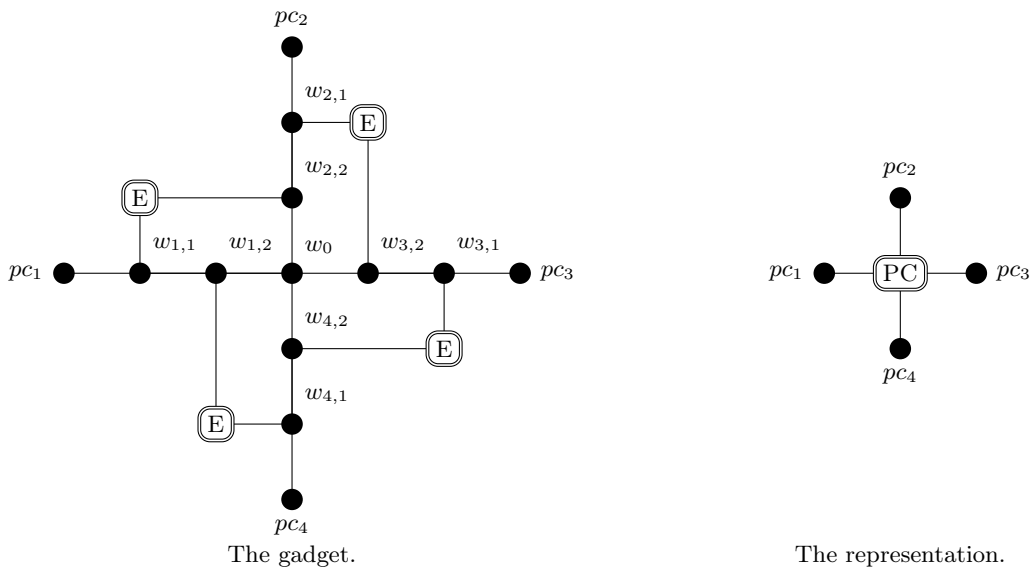

%% file: bifurcate.tex
% Bifurcate gadget
\begin{figure}[t!]
  \centering
  \begin{tikzpicture}[every node/.style={transform shape}, scale=0.5]
    \draw  (2,-1) \nodd{place2,} --  +(1,-1) \nodd{place2,} -- +(1,-3) \nodd{place2,} -- +(0,-4) \nodd{place2,} -- +(-1, -2) \nodd{place2,label=left:\huge{$a_i$}} -- +(0,0) -- +(0,-4);
    \draw  (2,-7) \nodd{place2,} --  +(1,-1) \nodd{place2,} -- +(1,-3) \nodd{place2,} -- +(0,-4) \nodd{place2,} -- +(-1, -2) \nodd{place2,label=left:\huge{$b_i$}} -- +(0,0) -- +(0,-4);
    \draw  (2,-13) \nodd{place2,} --  +(1,-1) \nodd{place2,} -- +(1,-3) \nodd{place2,} -- +(0,-4) \nodd{place2,} -- +(-1, -2) \nodd{place2,label=left:\huge{$c_i$}} -- +(0,0) -- +(0,-4);
    \draw (5,0) \nodd{place2,label=above:\huge{$a_{i,1}$}} -- +(-1,-1) \nodd{place2,} -- +(1,-1) \nodd{place2,} --cycle;
    \draw (8,0) \nodd{place2,label=above:\huge{$b_{i,1}$}} -- +(-1,-6) \nodd{place2,} -- +(1,-6) \nodd{place2,} --cycle;
    \draw (11,0) \nodd{place2,label=above:\huge{$c_{i,1}$}} -- +(-1,-12) \nodd{place2,} -- +(1,-12) \nodd{place2,} --cycle;
    \draw (4,-1) -- +(-1,-1) -- +(2,0);
    \draw (7,-6) -- +(-4,-2) -- +(2,0);
    \draw (10,-12) -- +(-7,-2) -- +(2,0);
    \draw (5,-2.5) \nodd{place2,} -- +(0,-3) \nodd{place2,} -- +(0,0) -- +(-2,-1.5) -- +(0,-3) -- (13,-4) \nodd{place2,label=20:\huge{$a_{i,2}$}} --cycle;
    \draw (5,-8.5) \nodd{place2,} -- +(0,-3) \nodd{place2,} -- +(0,0) -- +(-2,-1.5) -- +(0,-3) -- (13,-10) \nodd{place2,label=20:\huge{$b_{i,2}$}} --cycle;
    \draw (5,-14.5) \nodd{place2,} -- +(0,-3) \nodd{place2,} -- +(0,0) -- +(-2,-1.5) -- +(0,-3) -- (13,-16) \nodd{place2,label=20:\huge{$c_{i,2}$}} --cycle;
    \capt{5,-20}{\huge{(a) The gadget.}};

    \draw[line width = 1.5] (18,-1) --  +(1,-1) -- +(1,-3) -- +(0,-4) --cycle;
    \draw[line width = 1.5] (18,-7)  -- +(0,-4)  -- +(-1, -2) --cycle;
    \draw[line width = 1.5] (18,-13)  -- +(0,-4)  -- +(-1, -2) --cycle;
    \draw[line width = 1.5] (21, 0) -- +(1,-1) -- + (-1, -1) --cycle ;
    \draw[line width = 1.5] (21, -2.5) -- + (0, -3) -- (29,-4) --cycle;
    \draw[line width = 1.5] (19, -8) -- + (4, 2) -- + (6, 2) --cycle;
    \draw[line width = 1.5] (21, -8.5) -- + (-2, -1.5) -- + (0, -3) -- (29,-10) --cycle;
    \draw[line width = 1.5] (19, -14) -- + (7, 2) -- + (8.85, 2) --cycle;
    \draw[line width = 1.5] (21, -14.5) -- + (-2, -1.5) -- + (0, -3) --cycle;

    \draw  (18,-1) \nodd{place2,} --  +(1,-1) \nodd{place2,} -- +(1,-3) \nodd{place2,} -- +(0,-4) \nodd{place2,} -- +(-1, -2) \nodd{place2,label=left:\huge{$a_i$}} -- +(0,0) -- +(0,-4);
    \draw  (18,-7) \nodd{place2,} --  +(1,-1) \nodd{place2,} -- +(1,-3) \nodd{place2,} -- +(0,-4) \nodd{place2,} -- +(-1, -2) \nodd{place2,label=left:\huge{$b_i$}} -- +(0,0) -- +(0,-4);
    \draw  (18,-13) \nodd{place2,} --  +(1,-1) \nodd{place2,} -- +(1,-3) \nodd{place2,} -- +(0,-4) \nodd{place2,} -- +(-1, -2) \nodd{place2,label=left:\huge{$c_i$}} -- +(0,0) -- +(0,-4);
    \draw (21,0) \nodd{place2,label=above:\huge{$a_{i,1}$}} -- +(-1,-1) \nodd{place2,} -- +(1,-1) \nodd{place2,} --cycle;
    \draw (24,0) \nodd{place2,label=above:\huge{$b_{i,1}$}} -- +(-1,-6) \nodd{place2,} -- +(1,-6) \nodd{place2,} --cycle;
    \draw (27,0) \nodd{place2,label=above:\huge{$c_{i,1}$}} -- +(-1,-12) \nodd{place2,} -- +(1,-12) \nodd{place2,} --cycle;
    \draw (20,-1) -- +(-1,-1) -- +(2,0);
    \draw (23,-6) -- +(-4,-2) -- +(2,0);
    \draw (26,-12) -- +(-7,-2) -- +(2,0);
    \draw (21,-2.5) \nodd{place2,} -- +(0,-3) \nodd{place2,} -- +(0,0) -- +(-2,-1.5) -- +(0,-3) -- (29,-4) \nodd{place2,label=20:\huge{$a_{i,2}$}} --cycle;
    \draw (21,-8.5) \nodd{place2,} -- +(0,-3) \nodd{place2,} -- +(0,0) -- +(-2,-1.5) -- +(0,-3) -- (29,-10) \nodd{place2,label=20:\huge{$b_{i,2}$}} --cycle;
    \draw (21,-14.5) \nodd{place2,} -- +(0,-3) \nodd{place2,} -- +(0,0) -- +(-2,-1.5) -- +(0,-3) -- (29,-16) \nodd{place2,label=20:\begin{huge}$c_{i,2}$ \end{huge}} --cycle;
    \capt{21,-20}{\huge{(b) An example.}};
  \end{tikzpicture}

  \caption{Bifurcate gadget: To keep planarity, there is a path-crossing gadget in each edge intersection.}
  \label{fig:CPEC}
\end{figure}
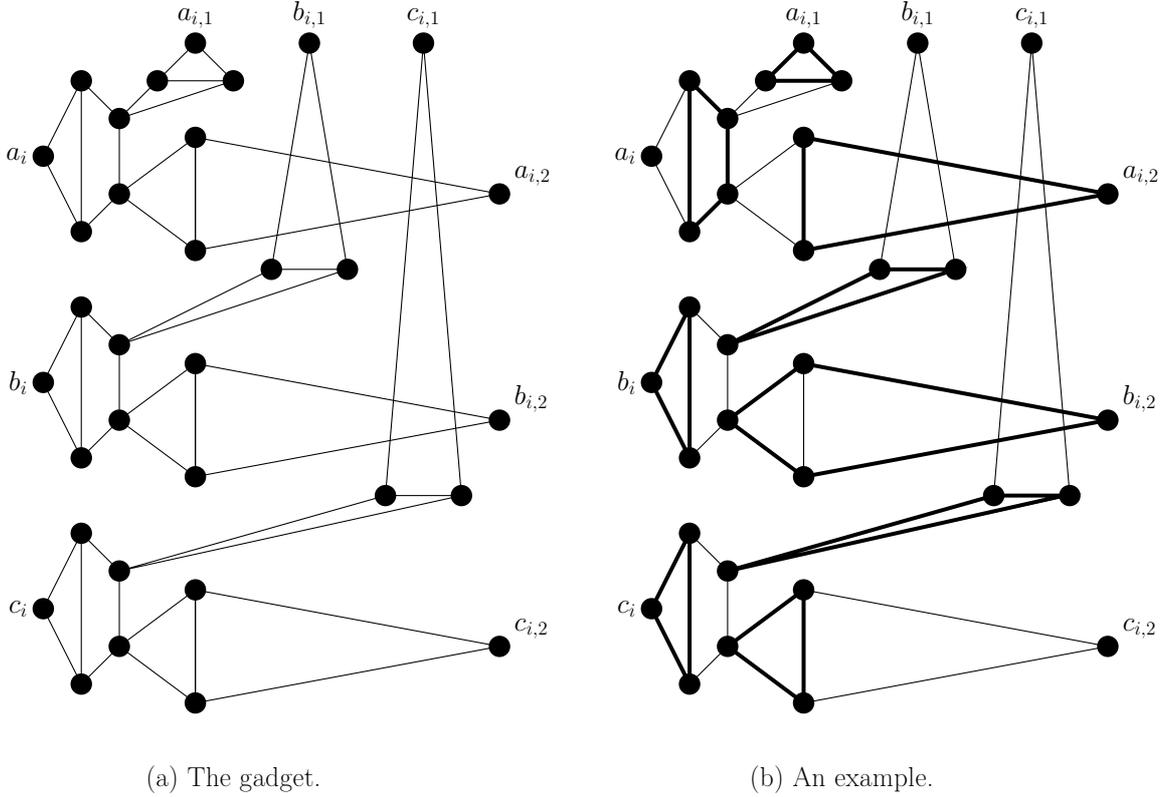

%% file: edge-gadget.tex
% Edge gadget
\begin{figure}[t!]
  \centering
  \begin{tikzpicture}[every node/.style={transform shape}, scale=0.5]
    \draw (8,-3) \nodd{place2,} -- +(0,-2) \nodd{place2,} -- +(0,0) -- +(-5,-1) \nodd{place2,label=left:\huge{$a_{i,k}$}} -- +(0,-2) -- (12,-11) \nodd{place2,label=right:\huge{$a_{j,k}$}} --cycle;
    \draw (4,-7) \nodd{place2,} -- +(0,-2) \nodd{place2,} -- +(0,0) -- +(-1,-1) \nodd{place2,label=left:\huge{$b_{i,k}$}} -- +(0,-2) -- (12,-8) \nodd{place2,label=right:\huge{$b_{j,k'}$}} --cycle;
    \draw (4,-10) \nodd{place2,} -- +(0,-2) \nodd{place2,} -- +(0,0) -- +(-1,-1) \nodd{place2,label=left:\huge{$c_{i,k}$}} -- +(0,-2) -- (12,-4) \nodd{place2,label=right:\huge{$c_{j,k'}$}} --cycle;
  \end{tikzpicture}

  \caption{Edge gadget: To keep planarity, there is a path-crossing gadget in each edge intersection.}
  \label{fig:CPEG}
\end{figure}
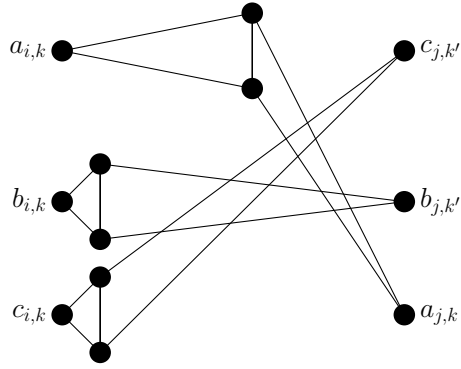

%% file: type3.tex
In this section we prove that the \textsc{Monochromatic Disjoint Paths} problem is of Type~3. We first need to introduce some definitions.
Let $G=(V,E)$ be a graph, let $k$ be an integer, and let $c : V \rightarrow \{0, \dots, k\}$ be a color function. Two colors $c_1$ and $c_2$ in $\{0, \dots, k\}$ are \emph{compatible}, and we denote it by $c_1 \equiv c_2$, if $c_1 = 0$, $c_2 = 0$, or $c_1 = c_2$. A path $P = x_1 \dots x_m$ in $G$ is \emph{monochromatic} if for all $i,j \in [m]$, $i \not = j$, $c(x_i)$ and $c(x_j)$ are two compatible colors. We let $c(P) = \max_{i \in [m]} (c(x_i))$.
We say that $P$ is \emph{colored} $x$ if $x = c(P)$.
Two monochromatic paths $P$ and $P'$ are \emph{color-compatible} if $c(P) \equiv c(P')$.

\probls
{\textsc{Monochromatic Disjoint Paths}}
{A graph $G=(V,E)$ of treewidth $\tw$, a color function $\gamma : V \rightarrow \{0, \dots, \tw\}$, an integer $m$, and a set $\NN = \{\NN_i = \{s_i, t_i\} | i \in [m], s_i,t_i \in V\}$.}
{Does $G$ contain $m$ pairwise vertex-disjoint monochromatic paths from $s_{i}$ to $t_{i}$, for $i \in [m]$?}
{The treewidth $\tw$ of $G$.}

%\subsection{Algorithm for Monochromatic Disjoint Paths}

The proof of the following lemma is inspired from the algorithm given in~\cite{Sch94} for the {\sc Disjoint Paths} problem on general graphs.

\begin{lemma}\label{lem:algoMonochromDisjointPaths}%$[\star]$
\textsc{Monochromatic Disjoint Paths} can be solved in time $2^{O(\tw \log \tw)} \cdot n^{O(1)}$.
\end{lemma}
\begin{proof} Again, we prove the lemma using branch-decomposition, which will lead the same asymptotic upper bounds in terms of the treewidth. Let $G$ be a colored graph and
let $\gamma: V(G) \rightarrow \{0,\dots, \tw\}$ be a coloring of $V(G)$.
Let $\{\NN_i=\{s_i,t_i\}\}_{i\in [m]}$ be the endvertices of the $m$ paths we are looking for, and
let $(T,\mu)$ a branch-decomposition of $G$ of width $\bw = \bw (G)$.
As in \cite{DPBF10}, we root $T$ by arbitrarily choosing an edge $e$ and subdivide it by inserting a new node $s$.
Let $e'$ and $e''$ be the new edges and set $\mids (e') = \mids (e'')= \mids (e)$.
We create a new root node $r$, connect it to $s$ by an edge $e_r$, and set $\mids(e_r) = \emptyset$. The root $e_r$ is not considered as a leaf.

Let now $e$ be an edge of $T$, let $X, P \subseteq \mids (e)$ with $X \cap P = \es$, and let $M, L$ be two disjoint matchings of $\mids (e) \bs (X \cup P)$. Let $\gamma_0: P \cup \set{M} \cup \set{L} \rightarrow \{0,\dots, \tw\}$ be a color function, and let $\varphi : P \rightarrow [m]$ be an injective function.
Intuitively, we want to keep track of the (partial) paths inside $G_e$, and to this end $P$ will correspond to the virtual sources of terminals, $M$ to the pairs of virtual sources to be linked by a path, $L$ to pairs of vertices $\{x,y\}$ such that there is a path in $G_e$ linking $x$ and $y$, and $X$ to vertices that are already inside a path or that are both an endpoint and a terminal.
We say that $mdp(G_e, \mids (e), X, P, M, L, \gamma_0,\varphi) = \true$ if the following conditions are fulfilled:

\begin{itemize}
\item[$\circ$] For all $\{s_i,t_i\}$ in $\NN \cap V(G_e)^2$,
\begin{itemize}
\item There exists a monochromatic path $s_i\dots t_i$ in $G_e$, or
\item There exist $\{s'_i, t'_i\} \in M$ and two monochromatic paths in $G_e$ $s_i\dots s'_i$ colored $\gamma_0(s'_i)$ and $t_i\dots t'_i$ colored $\gamma_0(t'_i)$ with $\gamma_0(s'_i) \equiv \gamma_0(t'_i)$.
\end{itemize}
\item[$\circ$] For all $\{s_i,t_i\}$ in $\NN$, such that $s_i \in V(G_e)$ and $t_i \not\in V(G_e)$ or vice-versa,
\begin{itemize}
\item There exist $s'_i \in P$ such that $\varphi(s'_i) = i$ and a monochromatic path $s_i=v_0\dots v_k=s'_i$ colored $\gamma_0(s'_i)$.
\end{itemize}
\item[$\circ$] For all $\{x_i,y_i\}$ in $L$,
\begin{itemize}
\item There exists in $G_e$ a monochromatic path $x_i\dots y_i$ colored $\max (\gamma_0(x_i),\gamma_0(y_i) )$.
\end{itemize}
\item[$\circ$] All these paths are vertex-disjoint and all vertices in $\mids (e)$ with degree at least 2 are in $X$.
\end{itemize}

Let $S_1 = (X_1, P_1, M_1, L_1, \gamma_1, \varphi_1)$ and $S_2 = (X_2, P_2, M_2, L_2, \gamma_2, \varphi_2)$ with $X_1, X_2, P_1,P_2, \dots $ defined as above. We define $G[S_1] = (P_1\cup \set{M_1} \cup \set{L_1}, \{\{x,y\} \in  L_1\})$ and colored by $\gamma_1$, and we define $G[S_2]$ analogously.
We say that $G[S_1, S_2]$ is \emph{defined} if for all $x\in V(G[S_1]) \cap V(G[S_2])$, $\gamma_1(x) \equiv \gamma_2(x)$, $X_1 \cap V(G[S_2]) = X_2 \cap V(G[S_1]) = X_1 \cap X_2 = \es$, and we define $G[S_1, S_2] = G[S_1] \cup G[S_2]$ and colored by $\gamma_{12}$  such that for all $x \in V(G[S_1,S_2])$, $\gamma_{12} = \max (\gamma_1(x), \gamma_2 (x))$.
Otherwise, we say that $G[S_1,S_2]$ is \emph{undefined}.

For each $e \in E(T)$, we define $\RR_e = \{(X, P, M, L, \gamma, \varphi) | X \subseteq \mids (e), P \subseteq \mids (e)$, $X \cap P = \es$, $M$ and $L$ are disjoint matchings on $ \mids (e) \bs (X \cup P)$, $\set{M} \cap \set{L} = \es$ \ and\ $mdp(G_e, \mids (e), X,P, M, L, \gamma, \varphi) = \true $. We want to know whether $(\es, \es, \es,\es, \es,\es) \in \RR_{e_r}$.  For each $e \in E(T)$, we can compute $\RR_e$ as follows:

\begin{itemize}
\item[$\circ$] if $e$ is a leaf, then $G_e = (\{x,y\}, \{(x,y)\}$, and
\begin{itemize}
\item if $\{x,y\} \in \NN$, then \\
$\RR_e = \{(\{x,y\}, \emptyset, \emptyset, \emptyset, \emptyset, \emptyset)\}$.
\item if $x \in \NN_i, y \in \NN_j$, $i \not = j$, then \\
$\RR_e = \{(\emptyset,  \{x,y\}, \emptyset, \emptyset,  \{(x,\gamma(x)),(y,\gamma(y))\}, \{(x,i), (y,j)\})\}$.
\item if $x \in \NN_i$ and $\forall j \in [m], y \not\in \NN_j$ and $\gamma(x) \not\equiv \gamma(y)$, then \\
$\RR_e = \{ (\es,\{x\}, \es, \es,  \{(x,\gamma(x)\}, \{(x,i)\})\}$.
\item if $x \in \NN_i$ and $\forall j \in [m], y \not\in \NN_j$ and $\gamma(x) \equiv \gamma(y)$, then
$\RR_e = \{ (\es,\{x\}, \es, \es,  \{(x,\gamma(x)\}, \{(x,i)\}),(\{x\}, \{y\}, \es, \es, \{(y,\max (\gamma(x), \gamma(y)))\}, \{(y,i)\})\}$.

\end{itemize}
\item[$\circ$] if $e$ is not a leaf, let $e_1$ and $e_2$ be the two children of $e$ in $E(T)$.
We construct $\RR_e$ as the set of all 6-tuples $(X, P, M, L, \gamma_0, \varphi)$ such that there exist  $S_1 = (X_1, P_1, M_1, L_1, \gamma_1, \varphi_1) \in \RR_{e_1}$ and $S_2 = (X_2, P_2, M_2, L_2, \gamma_2, \varphi_2) \in \RR_{e_2}$ fulfilling the following properties:

\begin{itemize}
\item $H = G[S_1, S_2]$ is defined;

\item For all $\{x_i, y_i\} \in L$, there exists a monochromatic path $x_i \dots y_i$ in $H$ and we have $\gamma_0 (x_i) = \gamma_0 (y_i) = \gamma_{12} (x_i \dots y_i)$;
\item All vertices in $\mids (e)$ of degree at least 2 in $G[S_1,S_2]$ are in $X$;
\item For all $\{v,w\} \in M_i$, $i \in \{1,2\}$, there is a monochromatic color-compatible path from $v$ to $w$ in $G[S_1,S_2]$ or two vertices $\{v', w'\} \in M$, and two monochromatic color-compatible paths $v \dots v'$ and $w \dots w'$ with $\gamma_0 (v') = \gamma_0 (w') = \max (\gamma_{12} (v \dots v'), \gamma_{12} (w \dots w'))$;
\item For all $i \in \{1,2\}$ an for all $v$ in $P_i$, there exist $w \in P$ and a monochromatic color-compatible path  $v \dots w$, or there exist $w \in P_{3-i}$ such that $\varphi_i(v) = \varphi_{3-i}(w)$ and a monochromatic path $v \dots w$ such that $\gamma_0(w) = \gamma_{12} (v\dots w)$;
\item All these paths are pairwise vertex-disjoint.
\end{itemize}
\end{itemize}

As in the graph $G[S_1,S_2]$ by construction all vertices have degree at most two, we can easily check all the previous properties in polynomial time, as we just have to compare two sets or traverse a path in $G[S_1,S_2]$ to verify each property. Therefore, we can compute each element of $\RR_e$ in time $\mbox{poly}(\mids (e))$. As $(X,P,\set {M}, \set {L})$ forms a partition of a subset of $\mids (e)$, there are at most $5^{\mids(e)}$ such 4-tuples. There are at most $\tw+1$ colors and at most $(\tw +1)^{\mids (e)}$ choices for $\gamma_0$.
As $|\{\varphi(x) | x\in P\}| \leq |P| \leq \mids (e)$, there are at most $\mids (e) ^ {\mids (e)}$ possible different color functions $\varphi$.
As $\bw -1 \leq \tw$ we have that for all $e$ in $E(T)$, $|\mids (e)| \leq \tw +1$, hence
 for all $e$ in $E(T)$, $ |\RR_e| \leq 5^{\tw+1} \cdot (\tw +1)^{2 (\tw+1)}$.
As for each $e \in E(T)$ such that $e$ is not a leaf, we have to merge the tables of the two children $e_1$ and $e_2$ of $e$, the above dynamic programming algorithm can solve \textsc{Monochromatic Disjoint Paths} in time $O(25^{\tw+1}\cdot (\tw +1)^{4 (\tw+1)} \cdot |V(G)|)$. Again, we note that  the constant can probably be optimized by using fast matrix multiplication~\cite{Wil12}.\end{proof}

%\subsection{Lower bound for Planar Monochromatic Disjoint Paths }

%\ig{Make a remark about the pathwidth}

\vspace{.3cm}

We need to define the \textsc{$k \times k$-Hitting Set} problem, first introduced in~\cite{LMS11a}.

%We do not use the original version defined in~\cite{LMS11a} but the one considered in~\cite{CNP11}:

\probls
{\textsc{$k \times k$-Hitting Set}}
{A family of sets $S_1,S_2, \dots,  S_m \subseteq [k]\times[k]$, such that each set contains at most one element from each row of $[k]\times[k]$.}
{Is there a set $S$ containing exactly one element from each row such that $S\cap S_i \not = \es$ for any $1 \leq i \leq m$?}
{$k$.}

\begin{theorem}[Lokshtanov \emph{et al}.~\cite{LMS11a}]
\label{th:hittingSet}
{\textsc{$k \times k$-Hitting Set}} cannot be solved in time $2^{o(k \log k)} \cdot m^{O(1)}$  unless the ETH fails.
\end{theorem}

We state the following theorem in terms of the pathwidth of the input graph, and as any graph $G$ satisfies $\tw(G) \leq \pw(G)$,  it implies the same lower bound in the treewidth.

\begin{theorem}
\label{th:lbmdp}
\textsc{Planar Monochromatic Disjoint Paths} cannot be solved in time $2^{o(\pw \log \pw)} \cdot n^{O(1)}$  unless the ETH fails.%, where $\pw$ stands for the pathwidth of the input graph.
\end{theorem}
\begin{proof}
We reduce from \textsc{$k \times k$-Hitting Set}. Let $k$ be an integer and $S_1, S_2 , \dots,  S_m \subseteq [k]\times[k]$ such that each set contains at most one element from each row of $[k]\times[k]$.
We will first present an overview of the reduction with all the involved gadgets, and then we will provide a formal definition of the constructed planar graph $G$.

We construct a gadget for each row $\{r\} \times [k]$, $r \in [k]$, which selects the unique pair $p$ of $S$ in this row.
First, for each $r \in [k]$, we introduce two new vertices $s_r$ and $t_r$, a request $\{s_r,t_r\}$, $m+1$ vertices $v_{r,i}$, $i \in \{0, \dots, m\}$, and $m+2$ edges $\{e_{r,0} = (s_r,v_{r,0})\} \cup \{e_{r,i} = (v_{r,i-1},v_{r,i}) | i \in [m] \} \cup  \{e_{r,m+1} = (v_{r,m},t_r)\}$. That is, we have a path with $m+2$ edges between $s_r$ and $t_r$.

Each edge of these paths, except the last one, will be replaced with an appropriate gadget.
Namely, for each $r \in [k]$, we replace the edge $e_{r,0}$ with the gadget depicted in Fig.~\ref{fig:CS}, which we call \emph{color-selection} gadget.
In this figure, vertex $u_{r,i}$ is colored $i$.
The color used by the path from $s_r$ to $t_r$ in the color-selection gadget will define the pair of the solution of $S$ in the row $\{r\} \times [k]$.

\input{type3gadget}

Now that we have described the gadgets that allow to define $S$, we need to ensure that $S\cap S_i \not = \es$ for any $i \in [m]$.
For this, we need the gadget depicted in Fig.~\ref{fig:EE}, which we call \emph{expel} gadget. Each time we introduce this gadget, we add to $\mathcal{N}$ the request $\{s,t\}$.
This new requested path uses either vertex $u$ or vertex $v$, so only one of these vertices can be used by other paths.
For each $i \in [m]$, we replace all the edges $\{e_{r,i} | r \in [k]\}$ with the gadget depicted in Fig. \ref{fig:S}, which we call \emph{set} gadget.
In this figure, $a_{r,i}$ is such that if $(\{r\} \times [k]) \cap S_i = \{\{r,c_{r,i}\}\}$ then $a_{r,i}$ is colored $c_{r,i}$, and if $(\{r\} \times [k]) \cap S_i = \es$ then vertex $a_{r,i}$ is removed from the gadget.

\input{FigSetGadget}

This completes the construction of the graph $G$, which is illustrated in Fig.~\ref{fig:FG}. Note that $G$ is indeed planar. Formally, the graph we obtain is $G = (V,E)$, where $V = \{s_r | r \in [k] \} \cup \{ t_r | r \in [k] \} \cup \{ v_{r,i} | r \in [k], i \in \{0,m\} \} \cup \{ u_{r,c} | r \in [k], c \in [k]\} \cup (\{ w_{r,i,b} | r \in [k], i \in [m], i \in \{1,2\} \} \bs \{w_{r,i,b} | i \in [m], (r,b) \in \{(1,1), (k,2)\}\} )\cup \{s_{r,i} | r \in [k-1], i \in [m]\} \cup \{ t_{r,i} | r \in [k-1], i \in [m]\} \cup \{ a_{r,i} | \exists c \in [k], (r,c) \in S_i$ and
$E = \{\{s_r, u_{r,c}\} \in V^2 | r \in [k], c \in [k] \} \cup \{\{u_{r,c}, v_{r,0}\} \in V^2 | r \in [k], c \in [k]\} \cup \{\{v_{r,i-1},w_{r,i,b}\} \in V^2 | r \in [k], i \in [m], b \in \{1,2\}\} \cup \{\{w_{r,i,b},v_{r,i}\} \in V^2 | r \in [k], i \in [m], b \in \{1,2\}\} \cup \{\{v_{r,i-1},a_{r,i}\} \in V^2 | r \in [k], i \in [m] \} \cup \{\{a_{r,i},v_{r,i}\} \in V^2 | r \in [k], i \in [m]\} \cup \{\{v_{r,m}, t_r\} \in V^2 | r \in [k]\} \cup \{\{s_{r,i}, w_{r,i,2} \} \in V^2 | r \in [k-1], i \in [m] \} \cup \{\{s_{r,i},w_{r+1,i,1}\} \in V^2 | r \in [k-1], i \in [m] \} \cup \{\{ t_{r,i}, w_{r,i,2}\} \in V^2 | r \in [k-1], i \in [m]\} \cup \{\{ t_{r,i},w_{r+1,i,1}\} \in V^2 | r \in [k-1], i \in [m]\}$.

\input{type3final}

The color function $\gamma$ of $G$ is defined such that for each $r \in [k]$ and $c \in [k]$, $\gamma(u_{r,c}) = c$, and for each $i \in [m]$ and $(r,c) \in S_i$, $\gamma(a_{r,i}) = c$. For any other vertex $v \in V(G)$, we set $\gamma (v) = 0$. Finally, the input of \textsc{Planar Monochromatic Disjoint Paths} is the planar graph $G$, the color function $\gamma$, and the $k+(k-1)\cdot m$ requests $\NN = \{\{s_r,t_r\} | r \in [k]\} \cup \{\{s_{r,i},t_{r,i}\}| r \in [k-1], i \in [m]\}$, the second set of requests corresponding to the ones introduced by the expel gadgets.

Note that because of the expel gadgets,  the request $\{s_r, t_r\}$ imposes a path between $v_{r,i-1}$ and $v_{r,i}$ for each $r \in [k]$.
Note also that because of the expel gadgets, at least one of the paths between $v_{r,i-1}$ and $v_{r,i}$ should use an $a_{r,i}$ vertex, as otherwise at least two paths would intersect.
Conversely, if one path uses a vertex $a_{r,i}$, then we can find all the desired paths in the corresponding set gadgets by using the vertices $w_{r,i,b}$.

Given a solution of  \textsc{Planar Monochromatic Disjoint Paths} in $G$,  we can construct a solution of \textsc{$k \times k$-Hitting Set} by letting $S = \{(r, c) | r \in [k]$ such that the path from $s_r$ to $t_r$ is colored with color $c \}$.
We have that $S$ contains exactly one element of each row, so we just have to check if $S \cap S_i \not = \es$ for each $i \in [m]$.
Because of the property of the set gadgets mentioned above, for each $i \in [m]$, the set gadget labeled $i$ ensures that $S \cap S_i \not = \es$.

Conversely, given a solution $S$ of  \textsc{$k \times k$-Hitting Set}, for each $\{r,c\} \in S$ we color the path from $s_r$ to $t_r$ with  color $c$. We assign an arbitrary coloring to the other paths.
For each $i \in [m]$, we take $\{r,c \} \in S \cap S_i$ and in the set gadget labeled $i$, we impose that the path from $v_{r, i-1}$ to $v_{r,i}$ uses  vertex $a_{r,i}$. By using the vertices $w_{r,i,b}$ for the other paths, we find the desired $k + (k-1)\cdot m$ monochromatic paths.

Let us now argue about the pathwidth of $G$. We define for each $r,c \in [k]$ the bag $B_{0,r,c} = \{s_{r'} | r' \in [k]\} \cup \{ v_{r',0} | r' \in [k]\} \cup \{ u_{r,c} \}$, for each $i \in [m]$, the bag $B_i = \{v_{r,i-1} | r \in [k]\} \cup \{v_{r,i} | r \in [k]\} \cup \{a_{r,i}\in V(G) | r \in [k]\} \cup \{w_{r,i,b}\in V(G) | r \in [k], b \in [2]\} \cup \{ s_{r,i} | r \in [m-1]\} \cup \{ t_{r,i} | r \in [m-1]\}$, and the bag $B_{m+1} = \{v_{r,m} | r \in [k] \} \cup \{ t_r | r\in [k]\}$. We note that the size of each bag is at most $ 2\cdot (k-1) + 5\cdot k-2 = O(k)$. A path decomposition of $G$ consists of all bags $B_{0,r,c}$, $r,c \in [k]$ and $B_i$, $i \in [m+1]$ and edges  $\{B_{i}, B_{i+1}\}$ for each $i \in [m]$, $\{B_{0,r,c},B_{0,r,c+1}\}$ for $r \in [k]$, $c\in [k-1]$, $\{B_{0,r,k}, B_{0,r+1,1}\}$ for $r \in [k]$, and $\{B_{0,k,k},B_1\}$. Therefore, as we have  that $\pw (G) = O(k)$,  if one could solve \textsc{Planar Monochromatic Disjoint Paths} in time $2^{o(\pw \log \pw)} \cdot n^{O(1)}$, then one could also solve {\textsc{$k \times k$-Hitting Set}} in time $2^{o(k \log k)} \cdot m^{O(1)}$, which is impossible by Theorem~\ref{th:hittingSet} unless the ETH fails.\end{proof}

%As any graph $G$ satisfies $\tw(G) \leq \pw(G)$, from Theorem \ref{th:lbmdp} we obtain the following corollary
%
%\begin{corollary}
%\textsc{Planar Monochromatic Disjoint Paths} cannot be solved in time $2^{o(\tw \log \tw)} \cdot n^{O(1)}$ unless the ETH fails.
%\end{corollary}

%% file: type3gadget.tex
%% Color selection
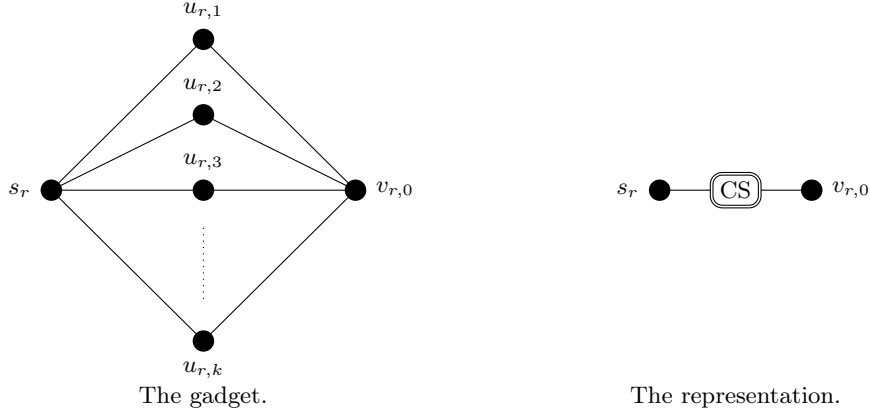
\begin{figure}[h]
  \centering
  \begin{tikzpicture}
    \draw (0,0)
    -- +(  2,  2) \nodd{label=above:$u_{r,1}$}
    -- +(  4,  0) ;
    \draw (0,  0)
    -- +(  2,  1) \nodd{label=above:$u_{r,2}$}
    -- +(  4,  0) ;
    \draw (0,  0)
    -- +(  2,  0) \nodd{label=above:$u_{r,3}$}
    -- +(  4,  0) ;
    \draw (0,  0) \nodd{label=left:$s_r$}
    -- +(+ 2,- 2) \nodd{label=below:$u_{r,k}$}
    -- +(  4,  0) \nodd{label=right:$v_{r,0}$};
    \draw[dotted] (2, -0.5) -- +(0,-1);
    \capt{1,-3}{The gadget.};
    \draw (8, 0) \nodd{label=left:$s_r$}
    -- +(1,0) \noddd{}{CS}
    -- +(2,0) \nodd{label=right:$v_{r,0}$};
    \capt{8,-3}{The representation.};
  \end{tikzpicture}

  \caption{Color-selection gadget, where $u_{r,i}$ is colored $c_i$ for each $i \in [k]$.}
  \label{fig:CS}
\end{figure}

%% Expel edge
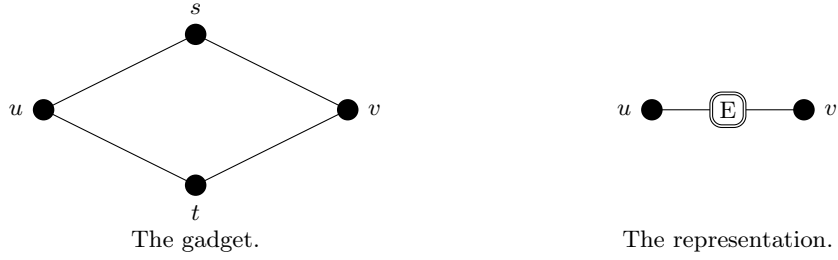
\begin{figure}[h]
  \centering
  \begin{tikzpicture}
    \draw (0,0) \nodd{label=above:$s$} -- +(2,-1) \nodd{label=right:$v$} -- +(0,-2) \nodd{label=below:$t$} -- +(-2,-1) \nodd{label=left:$u$} --cycle;
    \capt{-1,-3}{The gadget.};
    \draw (6,-1) \nodd{label=left:$u$} -- +(1,0) \noddd{}{E} -- +(2,0) \nodd{label=right:$v$};
    \capt{6,-3}{The representation.};
  \end{tikzpicture}

  \caption{Expel gadget.}
  \label{fig:EE}
\end{figure}

%% file: FigSetGadget.tex
%% Set gadget
\begin{figure}[t!]
  \vspace{-.4cm}
  \centering
  \begin{tikzpicture}
    \draw ( 0, 0) \nodd{label=left:$v_{1,i-1}$}
    -- +  ( 2, 0) \nodd{label=above:$a_{1,i}$}
    -- +  ( 4, 0) \nodd{label=right:$v_{1,i}$}
    ;
    \draw ( 0, 0)
    -- +  ( 2,-1)
    -- +  ( 4, 0)
    ;
    \draw ( 0,-3) \nodd{label=left:$v_{2,i-1}$}
    -- +  ( 2, 0) \nodd{label=above:$a_{2,i}$}
    -- +  ( 4, 0) \nodd{label=right:$v_{2,i}$}
    ;
    \draw ( 0,-3)
    -- +  ( 2, 1)
    -- +  ( 4, 0)
    ;
    \draw ( 0,-3)
    -- +  ( 2,-1)
    -- +  ( 4, 0)
    ;
    \draw ( 0,-6) \nodd{label=left:$v_{3,i-1}$}
    -- +  ( 2, 0) \nodd{label=above:$a_{3,i}$}
    -- +  ( 4, 0) \nodd{label=right:$v_{3,i}$}
    ;
    \draw ( 0,-6)
    -- +  ( 2, 1)
    -- +  ( 4, 0)
    ;
    \draw ( 0,-6)
    -- +  ( 2,-1)
    -- +  ( 4, 0)
    ;
    \draw ( 0,-10) \nodd{label=left:$v_{k,i-1}$}
    -- +  ( 2, 0) \nodd{label=above:$a_{k,i}$}
    -- +  ( 4, 0) \nodd{label=right:$v_{k,i}$}
    ;
    \draw ( 0,-10)
    -- +  ( 2, 1)
    -- +  ( 4, 0)
    ;
    \draw ( 2, -1) \nodd{label=right:$w_{1,i,2}$} -- +(0, -0.5) \noddd{}{E} -- +(0,-1) \nodd{label=right:$w_{2,i,1}$};
    \draw ( 2, -4) \nodd{label=right:$w_{2,i,2}$} -- +(0, -0.5) \noddd{}{E} -- +(0,-1) \nodd{label=right:$w_{3,i,1}$};
    \draw[dotted] (2, -7) \nodd{label=right:$w_{3,i,2}$} -- +(0, -2) \nodd{label=right:$w_{k,i,1}$};

    \capt{1,-11}{The gadgets.};

    \draw ( 9,-3) \nodd{label=left:$v_{1,i-1}$}
    -- +  ( 2, 0) \nodd{label=right:$v_{1,i}$}
    ;
    \draw ( 9,-4) \nodd{label=left:$v_{2,i-1}$}
    -- +  ( 2, 0) \nodd{label=right:$v_{2,i}$}
    ;
    \draw ( 9,-5) \nodd{label=left:$v_{3,i-1}$}
    -- +  ( 2, 0) \nodd{label=right:$v_{3,i}$}
    ;
    \draw ( 9,-7) \nodd{label=left:$v_{k,i-1}$}
    -- +  ( 2, 0) \nodd{label=right:$v_{k,i}$}
    ;
    \draw[dotted] ( 10,-5) -- +  ( 0,-2);
    \draw (10,-7)  \noddd{}{SET$_i$};
    \draw ( 10,-3) \noddd{}{SET$_i$}
    -- +  ( 0,-1) \noddd{}{SET$_i$}
    -- +  ( 0,-2) \noddd{}{SET$_i$}
    ;
    \capt{9,-11}{The representation.};
  \end{tikzpicture}

  \caption{Set gadgets.}
  \label{fig:S}
    \vspace{-.4cm}
\end{figure}
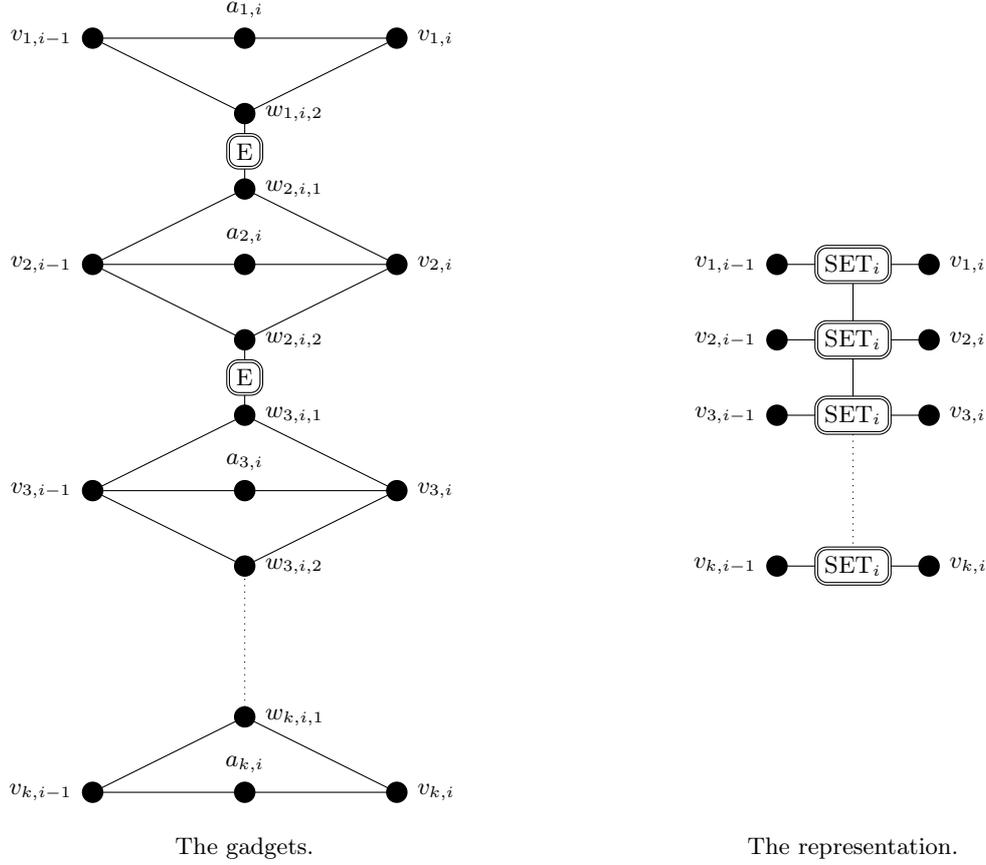

%% file: type3final.tex
%% Final graph
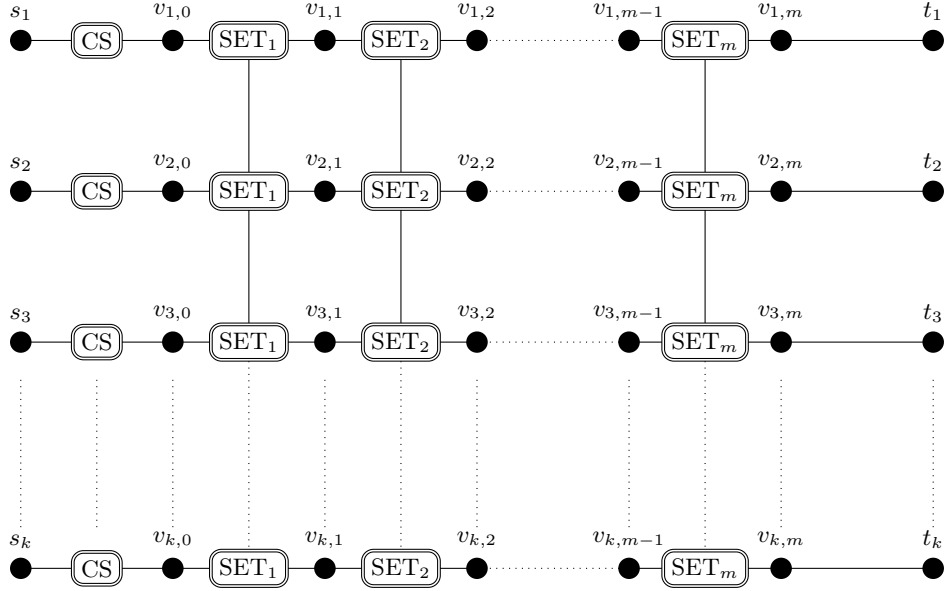
\begin{figure}[t]
  \centering
  \begin{tikzpicture}
    \draw ( 0, 0) \nodd{label=above:$s_1$}
    -- +  ( 1, 0) \noddd{}{CS}
    -- +  ( 2, 0) \nodd{label=above:$v_{1,0}$}
    -- +  ( 4, 0) \nodd{label=above:$v_{1,1}$}
    -- +  ( 6, 0) \nodd{label=above:$v_{1,2}$}
    ;
    \draw[dotted] (6,0) -- (8,0);
    \draw ( 8, 0) \nodd{label=above:$v_{1,m-1}$}
    -- +  (2, 0) \nodd{label=above:$v_{1,m}$}
    -- +  (4, 0) \nodd{label=above:$t_1$}
    ;

    \draw ( 0, -2) \nodd{label=above:$s_2$}
    -- +  ( 1, 0) \noddd{}{CS}
    -- +  ( 2, 0) \nodd{label=above:$v_{2,0}$}
    -- +  ( 4, 0) \nodd{label=above:$v_{2,1}$}
    -- +  ( 6, 0) \nodd{label=above:$v_{2,2}$}
    ;
    \draw[dotted] (6,-2) -- (8,-2);
    \draw ( 8, -2) \nodd{label=above:$v_{2,m-1}$}
    -- +  (2, 0) \nodd{label=above:$v_{2,m}$}
    -- +  (4, 0) \nodd{label=above:$t_2$}
    ;
    \draw ( 0, -4) \nodd{label=above:$s_3$}
    -- +  ( 1, 0) \noddd{}{CS}
    -- +  ( 2, 0) \nodd{label=above:$v_{3,0}$}
    -- +  ( 4, 0) \nodd{label=above:$v_{3,1}$}
    -- +  ( 6, 0) \nodd{label=above:$v_{3,2}$}
    ;
    \draw[dotted] (6,-4) -- (8,-4);
    \draw ( 8, -4) \nodd{label=above:$v_{3,m-1}$}
    -- +  (2, 0) \nodd{label=above:$v_{3,m}$}
    -- +  (4, 0) \nodd{label=above:$t_3$}
    ;

    \draw ( 0, -7) \nodd{label=above:$s_k$}
    -- +  ( 1, 0) \noddd{}{CS}
    -- +  ( 2, 0) \nodd{label=above:$v_{k,0}$}
    -- +  ( 4, 0) \nodd{label=above:$v_{k,1}$}
    -- +  ( 6, 0) \nodd{label=above:$v_{k,2}$}
    ;
    \draw[dotted] (6,-7) -- (8,-7);
    \draw ( 8, -7) \nodd{label=above:$v_{k,m-1}$}
    -- +  (2, 0) \nodd{label=above:$v_{k,m}$}
    -- +  (4, 0) \nodd{label=above:$t_k$}
    ;
    \draw[dotted] (3,-4) -- +(0,-3);
    \draw[dotted] (5,-4) -- +(0,-3);
    \draw[dotted] (9,-4) -- +(0,-3);
    \draw ( 3, 0) \noddd{}{SET$_1$}
    -- +  ( 0,-2) \noddd{}{SET$_1$}
    -- +  ( 0,-4) \noddd{}{SET$_1$}
    ;
    \draw ( 5, 0) \noddd{}{SET$_2$}
    -- +  ( 0,-2) \noddd{}{SET$_2$}
    -- +  ( 0,-4) \noddd{}{SET$_2$}
    ;
    \draw ( 9, 0) \noddd{}{SET$_m$}
    -- +  ( 0,-2) \noddd{}{SET$_m$}
    -- +  ( 0,-4) \noddd{}{SET$_m$}
    ;
    \draw ( 3, -7) \noddd{}{SET$_1$};
    \draw ( 5, -7) \noddd{}{SET$_2$};
    \draw ( 9, -7) \noddd{}{SET$_m$};
    \draw[dotted] (0,-4.5) -- +(0,-2);
    \draw[dotted] (1,-4.5) -- +(0,-2);
    \draw[dotted] (2,-4.5) -- +(0,-2);
    \draw[dotted] (4,-4.5) -- +(0,-2);
    \draw[dotted] (6,-4.5) -- +(0,-2);
    \draw[dotted] (8,-4.5) -- +(0,-2);
    \draw[dotted] (10,-4.5) -- +(0,-2);
    \draw[dotted] (12,-4.5) -- +(0,-2);
  \end{tikzpicture}

  \caption{Final graph $G$ in the reduction of Theorem \ref{th:lbmdp}.}
  \label{fig:FG}
\end{figure}

%% file: disjointPaths.tex
In this section we prove that, assuming the ETH, the \textsc{Planar Disjoint Paths} problem  cannot be solved in time $2^{o(\tw)} \cdot n^{O(1)}$.

\probls
{\textsc{Disjoint Paths}}
{A graph $G=(V,E)$, an integer $m$, and a set $\NN = \{\NN_i = \{s_i, t_i\} | i \in [m], s_i,t_i \in V\}$.}
{Does $G$ contain $m$ pairwise vertex-disjoint paths from $s_{i}$ to $t_{i}$, for $i \in [m]$?}
{The treewidth $\tw$ of $G$.}

\begin{theorem}\label{thm:DisjointPaths}
{\sc Planar Disjoint paths} cannot be solved in time $2^{o(\sqrt{n})} \cdot n^{O(1)}$ unless the ETH fails.
\end{theorem}

\input{disjoint_path_figure}

\begin{proof}
\label{th:lbdp}
We strongly follow the proof of Theorem \ref{th:lbcp}. Again, we reduce from {\sc Planar 3-Colorability} where the input graph has maximum degree at most 5.
Let $G=(V,E)$ be a planar graph with maximum degree at most 5 with $V = \{v_1, \dots, v_n\}$. We proceed to construct a planar graph $H$ together with a planar embedding. We construct the same graph as in the proof of Theorem \ref{th:lbcp} but where the gadgets are appropriately modified.
We reuse the  \emph{expel} gadget depicted in Fig. \ref{fig:EE} and for each expel gadget, we ask for a path between $s$ and $t$. We redefine  the \emph{double-expel} gadget as depicted in Fig. \ref{fig:DPE} and for each double-expel gadget, we ask for a path between $s$ and $t$. We reuse the \emph{path-crossing} gadget depicted in Fig. \ref{fig:CPPC} and only ask for the paths contained in the expel gadgets. We can now redefine the SC$_i$-gadget depicted in Fig. \ref{fig:DPSCG}, where each edge intersection is replaced with a path-crossing gadget. For each SC$_i$-gadget we ask for a path between $s_{i,0}$ and $t_{i,0}$. We also redefine the \emph{bifurcate} gadget as depicted in Fig. \ref{fig:DPEC}, and for each bifurcate gadget, we ask for a path between $s_{i,k}$ and $t_{i,k}$ for $k \in [9]$. Finally, we redefine the \emph{edge} gadget as depicted in Fig. \ref{fig:DPEG}, and for each edge gadget we ask for a path between $s_{i,j,k}$ and $t_{i,j,k}$ for $k \in [3]$. This completes the construction of the planar graph $H$. It can be easily checked that these gadgets preserve the same properties as the corresponding ones  in the proof of Theorem~\ref{th:lbcp}. Moreover, it is also easy to see that a path in a solution in $H$ cannot turn in a path-crossing gadget.

Given a solution of {\sc Planar Disjoint Paths}  in $H$,  for each $i \in [n]$ the selection of a cycle in the SC$_i$-gadget selects a color for $v_i$, that can be any common color in all color outputs of $v_i$, and the edge gadgets ensure that two adjacent vertices are in two different color classes. So in this way we obtain a solution of  \textsc{Planar 3-Colorability} in $G$.

Conversely, given a solution of \textsc{Planar 3-Colorability} in $G$,  it defines a color output for  $\{a_i, b_i, c_i\}$ for $i \in [n]$. Therefore, we select in the SC$_i$-gadget the path that uses the vertex in $\{a_i, b_i, c_i\}$ corresponding to the color of $v_i$. In each bifurcate gadget, we choose the paths that use the vertices in $\{a_{i,1}, b_{i,1}, c_{i,1}, a_{i,2}, b_{i,2}, c_{i,2}\}$ leading to two identical color outputs that coincide with the color output of $\{a_i, b_i, c_i\}$.
This choice satisfies the property that the color output of $\{a_i, b_i, v_i\}$ is contained in the color outputs of $\{a_{i,1}, b_{i,1}, c_{i,1}\}$ and $\{a_{i,2}, b_{i,2}, c_{i,2}\}$, and leaves as many free vertices as possible for other cycles in other gadgets.
Inside each edge gadget representing $\{v_i, v_j\} \in E$, we select the paths that are allowed by the free vertices. We complete our path selection by selecting a free path in each expel gadget contained in the path-crossing gadget.

As the degree of each vertex in $G$ is bounded by 5, the number of gadgets we introduce for each $v_i \in V(G)$ in order to construct $H$ is also bounded by a constant, so
the total number of vertices of $H$ is linear in the number of vertices of $G$.
Therefore, if we could solve {\sc Planar Disjoint Paths} in time $2^{o(\sqrt{n})}\cdot n^{O(1)}$, then we could also solve {\sc Planar 3-Colorability} in time $2^{o(\sqrt{n})} \cdot n^{O(1)}$, which is impossible by Theorem \ref{th:lb3c} unless the ETH fails.
The theorem follows.
\end{proof}\\

From Theorem~\ref{thm:DisjointPaths} we obtain the following corollary.

\begin{corollary}
\textsc{Planar Disjoint Paths} cannot be solved in time $2^{o(\tw)} \cdot n^{O(1)}$ unless the ETH fails.
\end{corollary}

%% file: disjoint_path_figure.tex
% Double-expel gadget
\begin{figure}[h]
  \vspace{-.3cm}
  \centering
  \begin{tikzpicture}%[every node/.style={transform shape}, scale=0.5]
    \draw ( 0, 0) \nodd{label=180:$u$}
    -- +  ( 0.5, 1) \nodd{label=90:$s$}
    -- +  ( 1, 0.5) \nodd{label=0:$u'$}
    -- +  ( 1,-0.5) \nodd{label=0:$u''$}
    -- +  ( 0.5,-1) \nodd{label=-90:$t$}
    --cycle;
  \end{tikzpicture}
  \caption{The double-expel gadget.}
  \label{fig:DPE}
\end{figure}
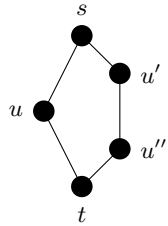

% SC_i gadget
\begin{figure}[t!]
  % \vspace{-.3cm}
  \centering
  \begin{tikzpicture}
    \draw (0,0) -- + (2,0) \nodd{label=right:$a_i$} -- + (0,-2);
    \draw (0,0) -- + (2,-1) \nodd{label=right:$b_i$} -- + (0,-2);
    \draw (0,0) -- + (2,-2) \nodd{label=right:$c_i$} -- + (0,-2);
    \draw (0,0) \nodd{label=90:$s_{i,0}$};
    \draw (0,-2) \nodd{label=-90:$t_{i,0}$};
  \end{tikzpicture}
  \caption{The SC$_i$-gadget: To keep planarity, there is a path-crossing gadget in each edge intersection.}
  \label{fig:DPSCG}
\end{figure}
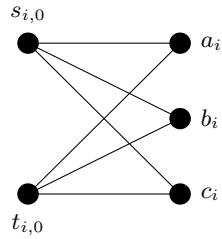

% Bifurcate gadget
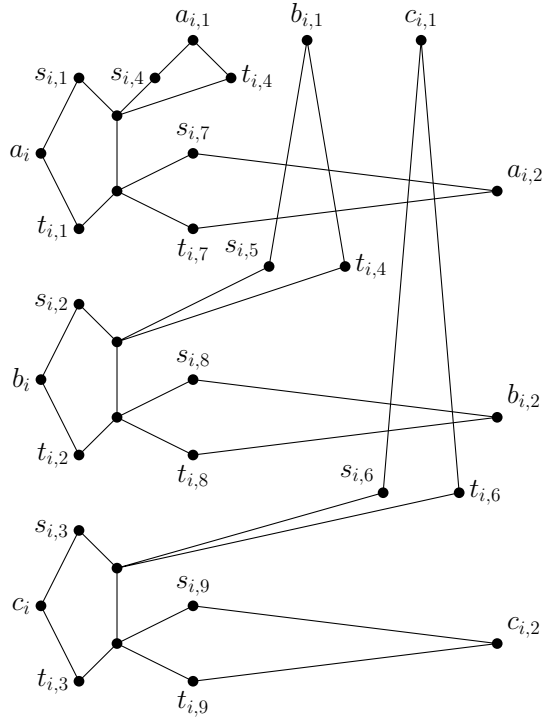
\begin{figure}[t!]
  \centering
  \begin{tikzpicture}[every node/.style={transform shape}, scale=0.5]
    \draw  (2,-1) \nodd{label=180:\huge{$s_{i,1}$}} --  +(1,-1) \nodd{} -- +(1,-3) \nodd{} -- +(0,-4) \nodd{label=180:\huge{$t_{i,1}$}} -- +(-1, -2) \nodd{label=left:\huge{$a_i$}} --cycle;
    \draw  (2,-7) \nodd{label=180:\huge{$s_{i,2}$}} --  +(1,-1) \nodd{} -- +(1,-3) \nodd{} -- +(0,-4) \nodd{label=-180:\huge{$t_{i,2}$}} -- +(-1, -2) \nodd{label=left:\huge{$b_i$}} --cycle;
    \draw  (2,-13) \nodd{label=180:\huge{$s_{i,3}$}} --  +(1,-1) \nodd{} -- +(1,-3) \nodd{} -- +(0,-4) \nodd{label=-180:\huge{$t_{i,3}$}} -- +(-1, -2) \nodd{label=left:\huge{$c_i$}} --cycle;
    \draw (5,0) \nodd{label=above:\huge{$a_{i,1}$}} -- +(-1,-1) \nodd{label=180:\huge{$s_{i,4}$}} -- +(-2,-2) -- +(1,-1) \nodd{label=0:\huge{$t_{i,4}$}} --cycle;
    \draw (8,0) \nodd{label=above:\huge{$b_{i,1}$}} -- +(-1,-6) \nodd{label=160:\huge{$s_{i,5}$}} -- +(-5,-8) -- +(1,-6) \nodd{label=0:\huge{$t_{i,4}$}} --cycle;
    \draw (11,0) \nodd{label=above:\huge{$c_{i,1}$}} -- +(-1,-12) \nodd{label=140:\huge{${s_{i,6}}$}} -- +(-8, -14) -- +(1,-12) \nodd{label=0:\huge{$t_{i,6}$}} --cycle;
    \draw (5,-3) \nodd{label=90:\huge{$s_{i,7}$}} --  +(-2,-1) -- +(0,-2) \nodd{label=-90:\huge{$t_{i,7}$}} -- (13,-4) \nodd{label=20:\huge{$a_{i,2}$}} --cycle;
    \draw (5,-9) \nodd{label=90:\huge{$s_{i,8}$}} --  +(-2,-1) -- +(0,-2) \nodd{label=-90:\huge{$t_{i,8}$}} -- (13,-10) \nodd{label=20:\huge{$b_{i,2}$}} --cycle;
    \draw (5,-15) \nodd{label=90:\huge{$s_{i,9}$}} -- +(-2,-1) -- +(0,-2) \nodd{label=-90:\huge{$t_{i,9}$}} -- (13,-16) \nodd{label=20:\huge{$c_{i,2}$}} --cycle;
  \end{tikzpicture}

  \caption{Bifurcate gadget: To keep planarity, there is a path-crossing gadget in each edge intersection.}
  \label{fig:DPEC}
\end{figure}

% Edge gadget
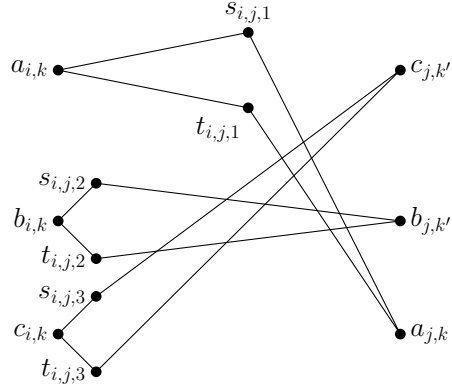
\begin{figure}[t!]
  \centering
  \begin{tikzpicture}[every node/.style={transform shape}, scale=0.5]
    \draw (8,-3) \nodd{label=90:\huge{$s_{i,j,1}$}} -- +(-5,-1) \nodd{label=left:\huge{$a_{i,k}$}} -- +(0,-2) \nodd{label=-125:\huge{$t_{i,j,1}$}}
    -- (12,-11) \nodd{label=right:\huge{$a_{j,k}$}} --cycle;
    \draw (4,-7) \nodd{label=180:\huge{$s_{i,j,2}$}} -- +(-1,-1) \nodd{label=left:\huge{$b_{i,k}$}} -- +(0,-2) \nodd{label=180:\huge{$t_{i,j,2}$}} -- (12,-8) \nodd{label=right:\huge{$b_{j,k'}$}} --cycle;
    \draw (4,-10) \nodd{label=180:\huge{$s_{i,j,3}$}} -- +(-1,-1) \nodd{label=left:\huge{$c_{i,k}$}} -- +(0,-2) \nodd{label=180:\huge{$t_{i,j,3}$}} -- (12,-4) \nodd{label=right:\huge{$c_{j,k'}$}} --cycle;
  \end{tikzpicture}

  \caption{Edge gadget: To keep planarity, there is a path-crossing gadget in each edge intersection.}
  \label{fig:DPEG}
\end{figure}

%% file: appendix.tex
\section{Proof of Theorem~\ref{th:lb3c}}
\label{sec:3coloring}

We start with defining some planar gadgets.
The first one is depicted in Fig. \ref{fig:C} and  called \emph{color gadget}, C-\emph{gadget} for short.
This gadget ensures that two vertices $u$ and $u'$ are in the same color class.
Note that we can extend the C-gadget for three vertices $u$, $u'$, and $u''$ and ensure the three vertices to be in the same color class by fixing a C-gadget between $u$ and $u'$ and another C-gadget between $u'$ and $u''$.
The second gadget is depicted in Fig. \ref{fig:CC} and called \emph{cross-color} gadget, CC-\emph{gadget} for short.
In this gadget, originally introduced in~\cite{GJE76}, one can check that if $u$,  $v$, $u'$, and $v'$ are in the same face before being connected by the gadget, and oriented in this order around the face, then $u$ and $u'$ are in the same color class and $v$ and $v'$ are in the same color class.

\input{type1gadget}

We reduce from \textsc{3-Colorability}. Let $G=(V,E)$ be an input general graph with $n = |V|$ and $V = \{v_1, \dots, v_n\}$, and we define the planar graph  $H$, illustrated in Fig.~\ref{fig:Hx} for $n=4$, as follows:
\begin{itemize}
\item[$\bullet$] For each $i \in [n]$, $u_{H,i}, v_{H,i}, w_{H,i} \in V(H)$;
\item[$\bullet$] For each $i,j \in [n]$, $i < j$, $\alpha_{H,i,j} \in V(H)$ and $\beta_{H,i,j} \in V(H)$;
\item[$\bullet$] For each $i \in \{1, \dots, n-1\}$, there is a C-gadget between $u_{H,i}$ and $\alpha_{H,i-1,i}$;
\item[$\bullet$] For each $i \in \{2, \dots, n\}$, there is a C-gadget between $u_{H,i}$ and $\beta_{H,i-1,i}$;
\item[$\bullet$] There is a C-gadget between $u_{H,n}$ and $w_{H,n}$;
\item[$\bullet$] There is a C-gadget between $u_{H,1}$ and $v_{H,1}$;
\item[$\bullet$] For each $i,j \in \{2, \dots, n-1\}$, $i<j$ there is a CC-gadget between $\alpha_{H,i,j}$, $\beta_{H,i,j}$, $\alpha_{H,i,j+1}$, and $\beta_{H,i-1,j}$;
\item[$\bullet$] For each $i \in \{2, \dots, n-1\}$, $i<j$ there is a CC-gadget between $\alpha_{H,i,n}$, $\beta_{H,i,n}$, $w_{H,i}$, and $\beta_{H,i-1,n}$;
\item[$\bullet$] For each $j \in \{2, \dots, n-1\}$, $i<j$ there is a CC-gadget between $\alpha_{H,1,j}$, $\beta_{H,1,j}$, $\alpha_{H,1,j+1}$, and $v_{H,j}$;
\item[$\bullet$] There is a CC-gadget between $\alpha_{H,1,n}$, $\beta_{H,1,n}$, $w_{H,1}$, and $v_{H,n}$;
\item[$\bullet$] For each $i,j \in [n]$, $i<j$, if $(v_{i},v_{j}) \in E$, then $(\alpha_{H,i,j}, \beta_{H,i,j}) \in E(H)$.
\end{itemize}

As the C-gadget and the CC-gadget are planar, $H$ is indeed planar (see Fig. \ref{fig:Hx}).

Because of the properties on the C-gadget and the CC-gadget, for each $i$ in $[n]$, $u_{H,i}$, $v_{H,i}$, and $w_{H,i}$ are in the same color class.
Because of the edges $(\alpha_{H,i,j}, \beta_{H,i,j})$, if there is an edge between $v_i$ and $v_j$ in $G$, then $u_{H,i}$ and $u_{H,j}$ should receive different colors.
If we have a 3-coloring of $G$, then by coloring $u_{H,i}$ with the color of $v_i$ for each $i \in [n]$, we find a 3-coloring of $H$. Conversely, if we have a coloring of $H$,  for each $i \in [n]$ we color each vertex $v_i$ of $V$ with the color of $u_{H,i}$.

Let us now argue about the maximum degree of the graph $H$. With the previous construction described so far, $H$ has maximum degree 7. In order to restrict it to 5, we replace each vertex of degree 6 or 7 with the two color gadgets shown in Fig.~\ref{fig:7to5}.

Let us finally argue about the number of vertices of $H$. Note that $H$ can be seen as a spanning subgraph of a grid of size $n$, where each vertex either has been replaced by a C-gadget or a CC-gadget, or it has been removed. As these two gadgets have at most $13$ vertices, and in the worst case, all these new vertices have degree 7 and we need to replace them with two color gadgets, that have 7 vertices each,  we have that $|V(H)| \leq 65\cdot n^2$. As {\sc 3-Colorability} cannot be solved in time $2^{o(n)}\cdot n^{O(1)}$ unless the ETH fails~\cite{IPZ01}, the theorem follows.

%% file: type1gadget.tex
% Color gadget
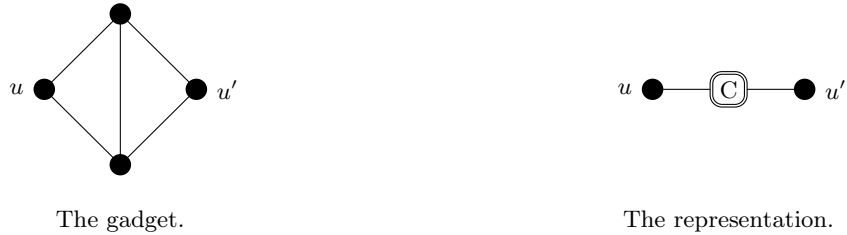
\begin{figure}[h]
  \centering
  \begin{tikzpicture}

    \draw (0,1) \nodd{} -- +(1,-1) \nodd{label=right:$u'$} -- +(0,-2) \nodd{} -- +(-1,-1) \nodd{label=left:$u$} -- +(0,0) -- +(0,-2);
    \capt{-1,-2}{The gadget.};

    \draw (7,0) \nodd{label=left:$u$} -- +(1,0) \noddd{}{C} -- +(2,0) \nodd{label=right:$u'$};
    \capt{7,-2}{The representation.};

  \end{tikzpicture}
  \vspace{-.15cm}
  \caption{Color gadget.}
  \vspace{-.15cm}
  \label{fig:C}
\end{figure}

\begin{figure}[h]
  \vspace{-.7cm}
  \centering
  \begin{tikzpicture}
    \nod{ 0}{ 0};
    \nod{-1}{ 0} edge ( 0, 0) edge (0,1) edge (0,-1) ;
    \nod{ 1}{ 0} edge ( 0, 0) edge (0,1) edge (0,-1) ;
    \nod{ 0}{ 1} edge ( 0, 0) ;
    \nod{ 0}{-1} edge ( 0, 0) ;
    \nod{ 1}{ 1} edge ( 1, 0) ;
    \nod{ 1}{-1} edge ( 0,-1) ;
    \nod{-1}{ 1} edge ( 0, 1) ;
    \nod{-1}{-1} edge (-1, 0) ;
    \nnod{ 0}{ 3}{$u$}{above} edge (1,1)  edge ( 0, 1) edge (-1, 1) ;
    \nnod{ 0}{-3}{$u'$}{below} edge (1,-1) edge ( 0,-1) edge (-1,-1) ;
    \nnod{ 3}{ 0}{$v'$}{right} edge (1,1)  edge ( 1, 0) edge ( 1,-1) ;
    \nnod{-3}{ 0}{$v$}{left } edge (-1,1) edge (-1, 0) edge (-1,-1) ;
    \capt{-1,-4}{The gadget.};
    \draw (8, 1) \nodd{label=above:$u$} -- +(0,-2) \nodd{label=below:$u'$};
    \draw (7, 0) \nodd{label=left:$v$} -- +(1,0) \noddd{}{CC}
    -- +(2,0) \nodd{label=right:$v'$};
    \capt{7,-4}{The representation.};
  \end{tikzpicture}
  \vspace{-.22cm}
  \caption{Cross-color gadget.}
   \vspace{-.22cm}

  \label{fig:CC}
\end{figure}
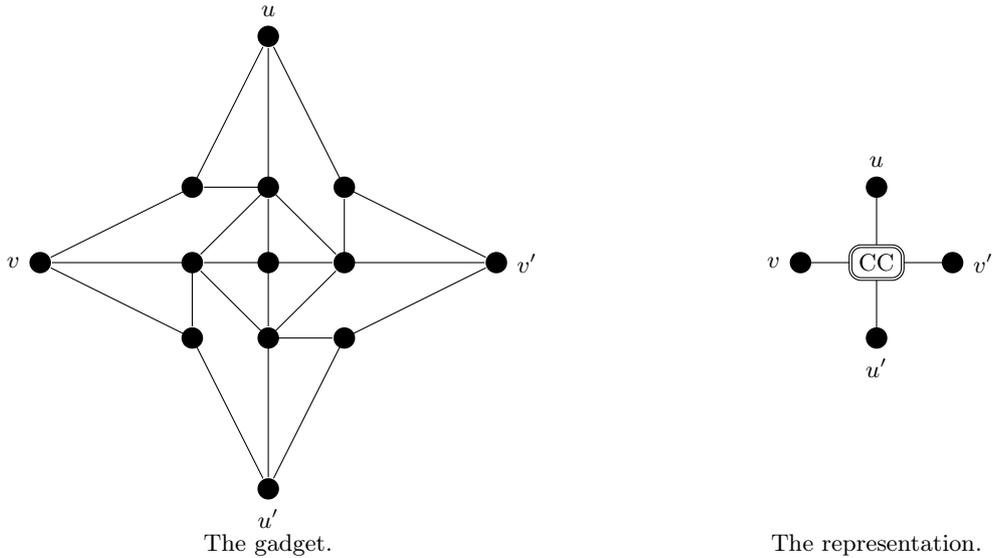

% The H_x graph
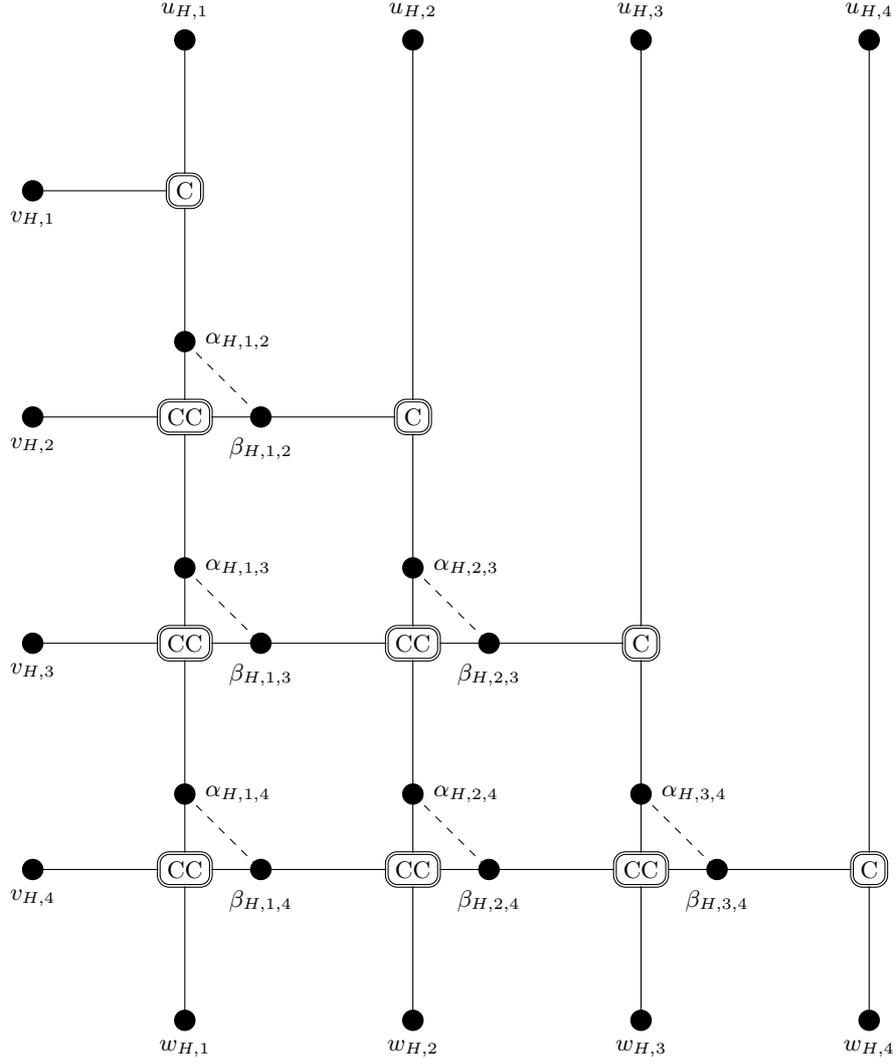
\begin{figure}[t!]
 \vspace{-.4cm}
  \centering
  \begin{tikzpicture}

    \node (n12) at (1,-2) [draw,double,rounded corners]{CC};
    \node (n13) at (1,-5) [draw,double,rounded corners]{CC};
    \node (n14) at (1,-8) [draw,double,rounded corners]{CC};
    \node (n23) at (4,-5) [draw,double,rounded corners]{CC};
    \node (n24) at (4,-8) [draw,double,rounded corners]{CC};
    \node (n34) at (7,-8) [draw,double,rounded corners]{CC};
    \node (c0) at (1,1) [draw,double,rounded corners]{C};
    \node (c1) at (4,-2) [draw,double,rounded corners]{C};
    \node (c2) at (7,-5) [draw,double,rounded corners]{C};
    \node (c3) at (10,-8) [draw,double,rounded corners]{C};
    \draw (1,3) node [place, label=above:$u_{H,1}$]{} -- (c0) -- +(0,-1) -- (n12) -- +(0,-1) -- (n13) -- +(0,-1) -- (n14) -- +(0,-2) node  [place, label=below:$w_{H,1}$]{};
    \draw (4,3) node [place, label=above:$u_{H,2}$]{} -- (c1) -- +(0,-1) -- (n23) -- +(0,-1)  -- (n24) -- +(0,-2) node  [place, label=below:$w_{H,2}$]{};
    \draw (7,3) node [place, label=above:$u_{H,3}$]{} -- (c2) -- (n34) -- +(0,-2) node  [place, label=below:$w_{H,3}$]{};
    \draw (10,3) node [place, label=above:$u_{H,4}$]{} -- (c3) -- +(0,-2) node  [place, label=below:$w_{H,4}$]{};

    \draw (-1,1) node [place, label=below:$v_{H,1}$] {} -- (c0);
    \draw (-1,-2) node [place, label=below:$v_{H,2}$] {} -- (n12) -- (c1);
    \draw (-1,-5) node [place, label=below:$v_{H,3}$] {} -- (n13) -- +(1,0) -- (n23) -- +(1,0) -- (c2);
    \draw (-1,-8) node [place, label=below:$v_{H,4}$] {} -- (n14) -- +(1,0) -- (n24) -- +(1,0) -- (n34) -- +(1,0) -- (c3);
    \draw[dashed]
    (1,-1) node [place,label=right:$\alpha_{H,1,2}$]{} --
    (2,-2) node [place,label=below:$\beta_{H,1,2}$] {};
    \draw[dashed]
    (1,-4) node [place,label=right:$\alpha_{H,1,3}$]{} --
    (2,-5) node [place,label=below:$\beta_{H,1,3}$] {};
    \draw[dashed]
    (1,-7) node [place,label=right:$\alpha_{H,1,4}$]{} --
    (2,-8) node [place,label=below:$\beta_{H,1,4}$] {};
    \draw[dashed]
    (4, -4) node [place,label=right:$\alpha_{H,2,3}$]{} --
    (5,-5) node [place,label=below:$\beta_{H,2,3}$] {};
    \draw[dashed]
    (4, -7) node [place,label=right:$\alpha_{H,2,4}$]{} --
    (5,-8) node [place,label=below:$\beta_{H,2,4}$] {};
    \draw[dashed]
    (7,-7) node [place,label=right:$\alpha_{H,3,4}$]{} --
    (8,-8) node [place,label=below:$\beta_{H,3,4}$] {};

  \end{tikzpicture}
  \caption{An example of graph $H$ for $n = 4$.}
  \label{fig:Hx}
\end{figure}

% degree 7 to degree 5
\begin{figure}[t]
\vspace{-.05cm}
  \centering
  \begin{tikzpicture}

    \draw (1,1) -- (0,0);
    \draw (1,1) -- (0,1);
    \draw (1,1) -- (0,2);
    \draw (1,1) -- (1,0);
    \draw (1,1) -- (2,0);
    \draw (1,1) -- (2,1);
    \draw (1,1) \nodd{} -- (2,2);

    \capt{0,-2}{The original vertex.};

    \draw (4,-1) -- (6,1);
    \draw (4,1) -- (6,1);
    \draw (4,3) -- (6,1);

    \draw (6,1) \nodd{} -- +(1,1) \nodd{} -- +(1,-1) -- +(1,1) -- +(2,0) \nodd{} -- +(1,-1) \nodd{} --cycle;
    \draw (8,1) \nodd{} -- +(1,1) \nodd{} -- +(1,-1) -- +(1,1) -- +(2,0) \nodd{} -- +(1,-1) \nodd{} --cycle;

    \draw (12,-1) -- (10,1);
    \draw (12,1) -- (10,1);
    \draw (12,3) -- (10,1);
    \draw (8,1) -- (8,-1);

    \capt{7,-2}{Simulation of the same vertex with maximum degree 5.};

  \end{tikzpicture}
  \caption{Reducing the maximum degree from 7 to 5.}
  \label{fig:7to5}
\end{figure}
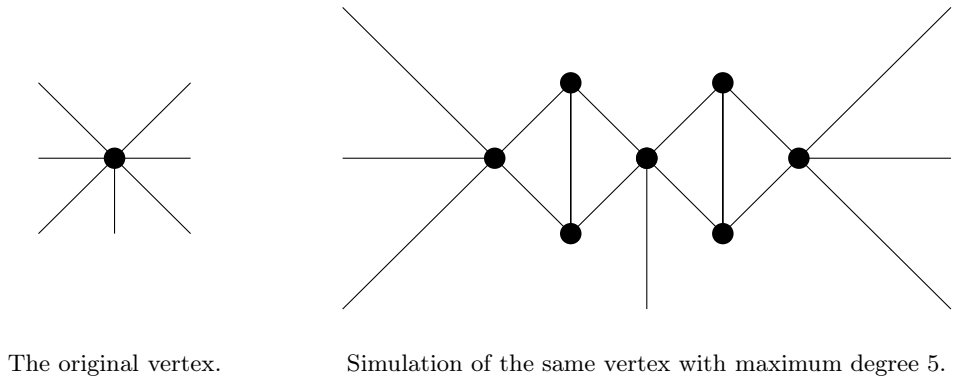